\documentclass[conference]{IEEEtran}

\usepackage{cite}      

\usepackage[usenames,dvipsnames,svgnames,table]{xcolor}
\usepackage{graphicx}  

\usepackage{algorithm}
\usepackage{algpseudocode}
\usepackage{multirow}

\usepackage{psfrag}    

\usepackage{url}       
\usepackage{tikz}
\usepackage{pgfplots}
\pgfplotsset{width=7cm,compat=1.3}

\usepackage{psfrag}
\usepackage{makecell}

\usepackage{amssymb}
\usepackage{mathtools}

\usepackage{amsmath}   
\usepackage{amssymb}

\usepackage{xfrac}

\newcommand{\nop}[1]{} 
\newcommand{\shorten}[1]{}
\newtheorem{proposition}{Proposition}
\newtheorem{theorem}{Theorem}
\newtheorem{definition}{Definition}

\newtheorem{lemma}{Lemma}

\newtheorem{corollary}{Corollary}
\newtheorem{example}{Example}

\usepackage{blindtext, subfig}
\usepackage{dblfloatfix} 

\newcommand{\signed}%
    {{\unskip\nobreak\hfill\penalty50
      \hskip2em\hbox{}\nobreak\hfil $\blacksquare$
      \parfillskip=0pt \finalhyphendemerits=0 \par}}
\newenvironment{proof}[1]
    {
    \bf{Proof:}\rm{\noindent{#1 }}\ignorespaces
    }
    {\signed\addvspace\medskipamount}

\begin{document}

\title{Network Traffic Driven Storage Repair}
%
\author{
	\IEEEauthorblockN{Danilo Gligoroski\IEEEauthorrefmark{1}, Katina Kralevska\IEEEauthorrefmark{1}, Rune E.~Jensen\IEEEauthorrefmark{2}, and Per Simonsen\IEEEauthorrefmark{3}\\ Email: \{danilog,katinak\}@ntnu.no, runeerle@idi.ntnu.no, per.simonsen@memoscale.com}
	\IEEEauthorblockA{\IEEEauthorrefmark{1}Dep. of Information Security and Communication Technology, NTNU, Norwegian University of Science and Technology}
	\IEEEauthorblockA{\IEEEauthorrefmark{2}Department of Computer Science, NTNU, Norwegian University of Science and Technology}
	\IEEEauthorblockA{\IEEEauthorrefmark{3}MemoScale AS, Norway}
}
\maketitle

\begin{abstract}	
Recently we constructed an explicit family of locally repairable and locally regenerating codes. Their existence was proven by Kamath et al. but no explicit construction was given. Our design is based on HashTag codes that can have different sub-packetization levels. In this work we emphasize the importance of having two ways to repair a node: repair only with local parity nodes or repair with both local and global parity nodes. We say that the repair strategy is network traffic driven since it is in connection with the concrete system and code parameters: the repair bandwidth of the code, the number of I/O operations, the access time for the contacted parts and the size of the stored file. We show the benefits of having repair duality in one practical example implemented in Hadoop. We also give algorithms for efficient repair of the global parity nodes.

\end{abstract}


%
\IEEEpeerreviewmaketitle

\section{Introduction}


The omnipresence of digitalization in the modern human civilization resulted in exponential growth of the digital universe. According to Cisco Global Cloud Index (CGI) \cite{networking2016cisco} the amount of all data stored by 2021 will be 7.2 ZB, with a proportion: 1.3 ZB stored in data centers, and 5.9 ZB stored in local devices. As a result of this, the importance of distributed storage systems rose to the level of being a backbone and a critical component for all existing infrastructures. Distributed storage systems became the crucial component for delivering IT services, providing storage services, enabling communications and networking to the users, devices and business processes \cite{networking2016cisco}. 

In order to provide a reliable service, distributed storage systems use a simple data replication (usually triplication). Replication becomes very expensive in terms of storage overhead due to the enormous amount of data stored in these data centers (measured with hundreds of petabytes). One alternative for providing reliability in data storage systems with significantly less overhead is to use erasure codes. Several big providers of distributed storage such as Windows Azure \cite{conf/usenix/HuangSXOCG0Y12} and Facebook Analytics Hadoop cluster \cite{journals/pvldb/SathiamoorthyAPDVCB13} have implemented different erasure codes. Recently, the official Apache distribution of Hadoop 3.0.0 \cite{ApacheHadoop} has started to give an option to use several classical Reed-Solomon erasure codes such as $(5,3)$, $(9, 6)$ and $(14, 10)$ codes in its file system HDFS. Reed-Solomon codes \cite{Reed:1960:PCC} are Maximum Distance Separable (MDS) codes, and thus they are optimal from the storage overhead point of view, but in practice they are expensive in the number of computations and in the amount of network traffic used in the recovery process. In 2010, Dimakis et al. \cite{5550492} showed the existence of another family of MDS codes - Minimum Storage Regenerating codes that minimize the amount of data transmitted over the network for repairing one failed node with MDS codes.

While the benefits of using erasure codes instead of simple replication are obvious in terms of the storage overhead, there are other aspects that are not that favorable for erasure codes. One of them is the so called \emph{repair efficiency}. In the case of a replication, the repair process is just a simple read (or copy) from the redundant data. On the other hand, the repair process with erasure coding involves data access from the non-failed nodes, transfer of the accessed data and decode computations at the node being repaired. It is essential to consider the concrete system and code parameters such as the repair bandwidth of the code, the number of I/O operations, the access time for the contacted parts and the size of the stored file when choosing the repair strategy with erasure codes. Arguably, we say the repair strategy is network traffic driven. In particular, there are two main metrics of the repair efficiency with erasure codes in distributed storage systems: the amount of transferred data during a repair process (\emph{repair bandwidth}) and the number of accessed nodes in a repair process (\emph{repair locality}). \emph{Regenerating Codes} (RCs) \cite{5550492} and \emph{Locally Repairable Codes} (LRCs) \cite{journals/tit/GopalanHSY12},\cite{conf/infocom/OggierD11}\cite{6195703} are optimized erasure codes for each of these two metrics, respectively.

In our recent work \cite{8328506}, we combine the benefits of RCs and LRCs together in one code construction. We construct an explicit family of locally repairable and locally regenerating codes whose existence was proven in a recent work by Kamath et al. \cite{6846301} about codes with local regeneration. In that work, an existential proof was given, but no explicit construction was given. Our explicit family of codes is based on HashTag codes \cite{7463553,8025778}.
HashTag codes are MDS vector codes with different vector length $\alpha$ (also called a sub-packetization level) that achieve the optimal repair bandwidth of MSR codes or near-optimal repair bandwidth depending on the sub-packetization level.
We apply the technique of parity-splitting of HashTag codes in order to construct codes with locality in which the local codes are regenerating codes and which hence, enjoy both advantages of locally repairable codes as well as regenerating codes. It is observed in \cite{Rashmi:2014:HGF:2619239.2626325} that 98.08\% of the failures in Facebook's data-warehouse cluster that consists of thousands of nodes are single failures. Thus, we optimize the repair for single failures although HashTag codes provide repair bandwidth savings for multiple failures as it is reported in \cite{8025778}.

We also show (although just with a concrete example) that the bound on the size of the finite field where these codes are constructed, given in the work by Kamath et al. \cite{6846301} can be lower. The presented explicit code construction has a practical significance in distributed storage systems as it provides system designers with greater flexibility in terms of selecting various system and code parameters due to the flexibility of HashTag code constructions.

We discuss the repair duality and its importance. Repair duality is a situation of having two ways to repair a node: to repair it only with local parity nodes or repair it with both local and global parity nodes.   To the best of the authors' knowledge, this is the first work that discusses the repair duality and how it can be applied based on concrete system and code parameters. Results from a Hadoop implementation illustrate the benefits of repair duality.

We further optimize the repair efficiency by giving an algorithm for shuffling the data in the parity nodes in order to reduce the repair bandwidth for the global parities.

The paper is organized as follows. Section \ref{RW} presents related work, and Section \ref{mathPrel} presents mathematical preliminaries. In Section \ref{paritySplitting}, we describe a framework for explicit constructions of locally repairable and locally regenerating codes. The repair process is analyzed in Section \ref{repair} where we explain the repair duality. In Section \ref{global}, we give an algorithm for efficient repair of the global parity nodes. Experimental results of measurements in Hadoop are given in Section \ref{Hadoop}. Conclusions are summarized in Section \ref{summary}.

\section{Related Work}\label{RW}

Minimum Storage Regenerating (MSR) codes are an important class of RCs that minimize the amount of data stored per node, due to the MDS property, and the repair bandwidth. The reduction in the repair bandwidth is achieved on the cost of contacting $n-1$ nodes, i.e., increasing the locality. There have been several research directions for MDS storage codes. 
The tradeoff between the sub-packetization level and the repair bandwidth of MDS codes have been investigated in \cite{7463553}, \cite{8025778}, \cite{2017arXiv170908216S}, \cite{8262840},\cite{2018arXiv180603103K}.
In \cite{7902203}, Goparaju et al. presented a construction of MSR codes for optimal repair of the systematic nodes for any number of accessed nodes $d\in\{k+1,\ldots,n-1\}$. The reduced locality comes at the cost of increased sub-packetization level. Another approach for constructing MSR codes with low sub-packetization level for different number of helper nodes is presented in \cite{8267088}.
Explicit constructions of MDS codes, including the MSR point, for optimal repair of the systematic nodes can be found in \cite{7463553,8025778}. Itani et al. used fractional repetition codes to minimize the total system recovery cost for single and multiple failures under various dynamic scenarios \cite{ITANI201784,ITANI20181}. 

On the other hand, LRCs relax the MDS requirement in order to minimize the number of nodes accessed during a repair. Studies on implementation and performance evaluation of LRCs can be found in \cite{Huang:2013:PCF:2435204.2435207,conf/usenix/HuangSXOCG0Y12,journals/pvldb/SathiamoorthyAPDVCB13,7593121}. Locally repairable codes with multiple repair alternatives have been proposed in \cite{6620355,7593121}, while codes with unequal locality have been presented in \cite{7541336}.
Zhang et al. presented local erasure recovery scheme by parity splitting of Vandermonde matrices \cite{7556268} where the regular entries of a Vandermonde matrix enable low-complexity software implementations.
Another way of constructing LRCs for distributed storage is based on low-density parity-check (LDPC) codes \cite{8094003}. LDPC codes have an inherent local repair property as LRCs but their reliability increases with the code length that on the other hand has a negative impact on the computation and the buffer requirements.

Combining the benefits of RCs and LRCs in one storage system can bring huge savings in practical implementations. 
For instance, repair bandwidth savings by RCs are important when repairing huge amounts of data, while a fast recovery and an access to small number of nodes enabled by LRCs are desirable for repair of frequently accessed data. 
Several works present code constructions that combine the benefits of RCs and LRCs \cite{6655894,6846301,7282575}. 
Rawat et al. in \cite{6655894} and Kamath et al. in \cite{6846301} have independently investigated codes with locality in the context of vector codes, and they call them locally repairable codes with local minimum storage regeneration (MSR-LRCs) and codes with local regeneration, respectively. 
Rawat et al. \cite{6655894} provided an explicit construction, based on Gabidulin maximum rank-distance codes, of vector linear codes with all-symbol locality for the case when the local codes are MSR codes. However, the complexity of these codes increases exponentially with the number of nodes due to the two-stage encoding. 
In \cite{6846301}, Kamath et al. gave an existential proof without presenting an explicit construction. 
Another direction of combining RCs and LRCs is to use repair locality for selecting the accessed nodes in a RC \cite{7282575}, while an interpretation of LRCs as exact RCs was presented in \cite{7541379}.
Two different erasure codes, product and LRC codes, are used to dynamically adapt to the workload changes in Hadoop Adaptively-Coded Distributed File System (HACFS) \cite{188446}.
Though there are several proposals for combining two different types of storage codes, they are lacking some of the desired attributes (see Table \ref{compare}).
Readers interested in an extended overview of erasure codes for distributed storage are referred to \cite{2018arXiv180604437B}.

\renewcommand{\arraystretch}{1.5}
\begin{table*}
	\caption{Comparison of locally repairable and locally regenerating codes presented in this paper with other erasure codes for storage.
	\vspace{-0.3cm}
	} \label{compare}
	\begin{center}
		\begin{tabular}{|l|l|l|l|l|l|}
			\hline
			Code & Systematic & Construction type & Order of GF & Any sub-packetization & Optimal repair of global parities\\
			\hline
			MSR-LRCs \cite{6655894} & Yes & Explicit & High & No (only MSR point) & No\\
			\hline
			Codes with local regeneration \cite{6846301} & Yes & Existence & High & No (only MSR point) & No\\
			\hline
			HACFS \cite{188446} & Yes & Explicit & Low & No & Product codes: Yes, LRC: No\\
			\hline
			This paper & Yes & Explicit & Low & Yes & Yes \\
			\hline
		\end{tabular}
	\end{center}
\end{table*}


\section{Mathematical Preliminaries}\label{mathPrel}
Inspired by the work of Gopalan et al. about locally repairable codes \cite{journals/tit/GopalanHSY12}, Kamath et al. extended and generalized the concept of locality  in \cite{6846301}. In this paper, we use notation that is mostly influenced (and adapted) from those two papers.
\textbf{\textit{Notations}}. For two integers $0<i<j$, we denote the set $\{i,i+1,\ldots, j \}$ by $[i:j]$, while the set $\{1,2,\ldots, j \}$ is denoted by $[j]$. 
Vectors and matrices are denoted with a bold font.

\begin{definition}[\cite{6846301}]
	A $\mathbb{F}_q$-linear vector code of block length $n$ is a code $\mathcal{C} \in \big(\mathbb{F}_q^\alpha \big)^n$ having a symbol alphabet $\mathbb{F}_q^\alpha$ for some $\alpha \geq  1$, i.e., 
	$$\mathcal{C} = \{ \mathbf{c} = (\mathbf{c}_1, \mathbf{c}_2, \ldots, \mathbf{c}_n), \mathbf{c}_i \in \mathbb{F}_q^\alpha \text{ for all } i \in [n] \}, $$ 
	and satisfies the additional property that for given $\mathbf{c},\mathbf{c}' \in \mathcal{C}$ and $a,b \in \mathbb{F}_q$, 
	$$a \mathbf{c} + b\mathbf{c}' = (a \mathbf{c}_1 + b\mathbf{c}_1', a \mathbf{c}_2 + b\mathbf{c}_2', \ldots, a \mathbf{c}_n + b\mathbf{c}_n' )$$ also belongs to $\mathcal{C}$ where $a\mathbf{c}_i$ is a scalar multiplication of the vector $\mathbf{c}_i$. $\blacksquare$
\end{definition}

Throughout the paper, we refer to the vectors $\mathbf{c}_i$ as vector symbols or nodes. Working with systematic codes, it holds that for the systematic nodes $\mathbf{c}_i = \mathbf{d}_i$ for $1\leq i\leq k$ and for the parity nodes $\mathbf{c}_{k+i} = \mathbf{p}_i$ for $1\leq i\leq r$. For every vector code $\mathcal{C} \in \big(\mathbb{F}_q^\alpha \big)^n$ there is an associated scalar linear code $\mathcal{C}^{(s)}$ over $\mathbb{F}_q$ of length $N=\alpha n$. Accordingly, the dimension of the associated scalar code $\mathcal{C}^{(s)}$ is $K = \alpha k$. For a convenient notation, the generator matrix $\mathbf{G}$ of size ${K \times N}$ of the scalar code $\mathcal{C}^{(s)}$ is such that each of the $\alpha$ consecutive columns corresponds to one code symbol $\mathbf{c}_i, i \in [n]$, and they are considered as $n$ thick columns $\mathbf{W}_i,  i \in [n]$. For a subset $\mathcal{I} \subset [n]$ we say that it is an information set for $\mathcal{C}$ if the restriction $\mathbf{G} |_\mathcal{I}$ of $\mathbf{G}$ to the set of thick columns with indexes lying in $\mathcal{I}$ has a full rank, i.e., $\text{rank}(\mathbf{G}|_\mathcal{I}) = K$. 

The minimum cardinality of an information set is referred as quasi-dimension $\kappa$ of the vector code $\mathcal{C}$. As the vector code $\mathcal{C}$ is  $\mathbb{F}_q$-linear, the minimum distance $d_{\min}$ of $\mathcal{C}$ is equal to the minimum Hamming weight of a non-zero codeword in $\mathcal{C}$. Finally, a vector code of block length $n$, scalar dimension $K$, minimum distance $d_{\min}$, vector-length parameter $\alpha$ and quasi-dimension $\kappa$ is shortly denoted with $[n,K,d_{\min},\alpha,\kappa]$. While in the general definition of vector codes in \cite{6846301} the quasi-dimension $\kappa$ does not necessarily divide the dimension $K$ of the associated scalar, for much simpler and convenient description of the codes in this paper we take that $k=\kappa$, i.e., $K=\alpha \kappa$. In that case the erasure and the Singleton bounds are given by: 
\begin{equation}\label{SingletonVector}
d_{\min} \leq n - \kappa + 1.
\end{equation} 
In \cite{5550492}, Dimakis et al. studied the repair problem in a distributed storage system where a file of $M$ symbols from a finite field $\mathbb{F}_q$ is stored across $n$ nodes, each node stores $\frac{M}{k}$ symbols. They introduced the metric repair bandwidth $\gamma$, and proved that the repair bandwidth of a MDS code is lower bounded by
\begin{equation}
\gamma \geq \frac{M}{k}\frac{d}{d-k+1},
\label{optimalBW-bound}
\end{equation}
where $d$ is the number of accessed available nodes (helpers).

\begin{lemma}[\cite{5550492}] The repair bandwidth of a $(n, k)$ MDS code is minimized for $d=n-1$. MSR codes achieve the lower bound of the repair bandwidth equal to
	\begin{equation}
	\gamma_{MSR}^{min}=\frac{M}{k}\frac{n-1}{n-k}.
	\label{optimalMSR}
	\end{equation}
\end{lemma}
\nop{
\begin{lemma}[\cite{5550492}] A $(n, k)$ MSR code attains the minimum storage point of the optimal tradeoff curve between the storage and the repair bandwidth, i.e.,
	\begin{equation}
	(\alpha_{MSR}, \gamma_{MSR}^{min})=(\frac{M}{k}, \frac{M}{k} \frac{n-1}{n-k}),
	\label{optimalMSR}
	\end{equation}
	when $n-1$ helper nodes are contacted.
\end{lemma}
}
A $(n, k)$ MSR code has the maximum possible distance $d_{min}=n-k+1$ in addition to minimizing the repair bandwidth, but it has the worst possible locality.
\begin{corollary} The locality of a $(n, k)$ MSR code is equal to $n-1$.
\end{corollary}

Any $[n,K,d_{\min},\alpha,\kappa]$ vector code $\mathcal{C}$ is MDS if and only if its generator matrix can be represented in the form $\mathbf{G} = [\mathbf{I} | \mathbf{P}]$, where the $K \times (N - K)$ parity matrix 
		\begin{equation}\label{parityPart}
		\mathbf{P}=
		\begin{bmatrix}
		\mathbf{G}_{1, 1} & \mathbf{G}_{1, 2} & \ldots & \mathbf{G}_{1, \kappa}\\
		\mathbf{G}_{2, 1} & \mathbf{G}_{2, 2} & \ldots & \mathbf{G}_{2, \kappa}\\
		\vdotswithin{1} & \vdotswithin{\alpha_n} & \ddots & \vdotswithin{{\alpha_n}^{k-1}}\\
		\mathbf{G}_{\kappa, 1} & \mathbf{G}_{\kappa, 2} & \ldots & \mathbf{G}_{\kappa, n-\kappa}\\
		\end{bmatrix},
		\end{equation}	
possesses the property that every square block submatrix of $\mathbf{P}$ is invertible. The $\mathbf{G}_{i,j}$ entries are square sub-matrices of size $\alpha \times \alpha$, and a block submatrix is composed by different entries of $\mathbf{G}_{i,j}$.

In order to analyze codes with local regeneration, Kamath et al. introduced a new family of vector codes called uniform rank-accumulation (URA) codes in \cite{6846301}. They showed that exact-repair MSR codes belong to the class of URA codes.

\begin{definition}\cite[Def. 2]{6846301}
	\label{KamathInformationLocality}
	Let $\mathcal{C}$ be a $[n,K,d_{\min},\alpha,\kappa]$ vector code with a generator matrix $\mathbf{G}$. The code $\mathcal{C}$ is said to have $(l,\delta)$ information locality if there exists a set of punctured codes $\{\mathcal{C}_i \}_{i \in \mathcal{L}}$ of $\mathcal{C}$ with respective supports $\{S_i \}_{i \in \mathcal{L}}$ such that
	\begin{itemize}
		\item $|S_i| \leq l+\delta - 1,$
		\item $d_{\min}(\mathcal{C}_i)\geq \delta,$ and
		\item $\text{rank}(\mathbf{G}|_{\bigcup_{i \in \mathcal{L}}})=K.$
	\end{itemize}
\end{definition}

If we put $\delta=2$ in Def.\ref{KamathInformationLocality}, then we get the definition of information locality introduced by Gopalan et al. \cite{journals/tit/GopalanHSY12}. 
They derived the upper bound for the minimum distance of a $(n, k, d)_q$ code with information locality $l$ for $\delta=2$ as 
	\begin{equation}
	d_{\min} \leq n-k-\left\lceil \frac{k}{l} \right\rceil+2.
	\label{distance1}
	\end{equation}

A general upper bound was derived in \cite{6846301} as
	\begin{equation}
	d_{\min} \leq n-k+1-\left( \left\lceil \frac{k}{l} \right\rceil - 1 \right)(\delta - 1) .
	\label{distance}
	\end{equation}

Huang et al. showed the existence of Pyramid codes that achieve the minimum distance given in (\ref{distance1}) when the field size is big enough \cite{Huang:2013:PCF:2435204.2435207}.
Finally, based on the work by Gopalan et al. \cite{journals/tit/GopalanHSY12} and Pyramid codes by Huang et al. \cite{Huang:2013:PCF:2435204.2435207}, Kamath et al. proposed a construction of codes with local regeneration based on a parity-splitting strategy in \cite{6846301}. 
\shorten{
\begin{figure}
	\begin{minipage}[b]{0.5\linewidth}
		\centering
		\includegraphics[width=8.9cm,height=5cm]{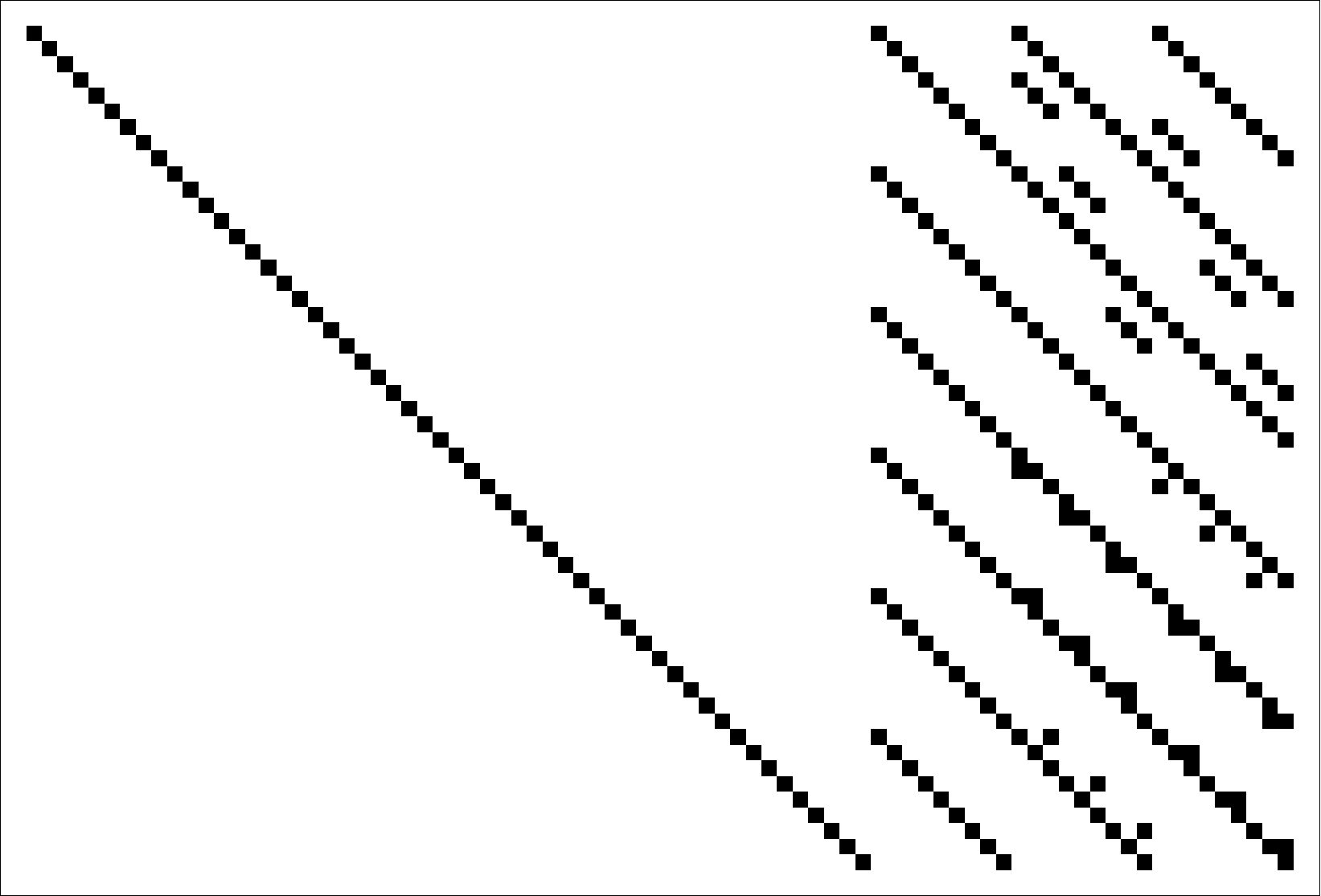}
	\end{minipage}
	\caption{A systematic generator matrix of the associated scalar code. The black squares on the main diagonal represent the value 1, but the black squares in the parity parts represent the non-zero values in $\mathbb{F}_{32}$. Note that if the parity matrix is partitioned in $9\times9$ square submatrices, it has the same form as in equation (\ref{parityPart}).}
	\label{G0906alpha09}
\end{figure}
}

\section{Codes with Local Regeneration from HashTag Codes by Parity-Splitting}\label{paritySplitting}

In \cite{7463553,8025778}, a new class of vector MDS codes called HashTag codes is defined. HashTag codes achieve the lower bound of the repair bandwidth given in (\ref{optimalMSR}) for $\alpha=r^{\lceil\frac{k}{r}\rceil}$, while they have near-optimal repair bandwidth for small sub-packetization levels. HashTag codes are of a great practical importance due to their properties: flexible sub-packetization level, small repair bandwidth, and optimized number of I/O operations.  
We briefly give the basic definition of HashTag codes before we construct codes with local regeneration from them by using the framework of parity-splitting discussed in \cite{6846301}.

\begin{definition}\label{HashTagCodes}
	A $(n,k,d)_q$ HashTag linear code is a vector systematic code defined over an alphabet $\mathbb{F}_q^\alpha$ for some $\alpha \geq  1$. It encodes a vector $\mathbf{x} = (\mathbf{x}_1,\ldots,\mathbf{x}_k)$, where $\mathbf{x}_i = (x_{1,i}, x_{2,i},\ldots,x_{\alpha,i})^T \in \mathbb{F}_q^\alpha$ for $i \in [k] $, to a codeword $\mathcal{C}(\mathbf{x}) = \mathbf{c} = (\mathbf{c}_1, \mathbf{c}_2, \ldots, \mathbf{c}_n)$ where the systematic parts $\mathbf{c}_i=\mathbf{x}_i$ for $i \in [k]$ and the parity parts $\mathbf{c}_i=(c_{1,i}, c_{2,i},\ldots,c_{\alpha,i})^T$ for $i \in [k:n]$ are computed by the linear expressions that have a general form as follows:
	\begin{equation}\label{LinEquations}
	c_{j,i}=\sum f_{\nu,j, i} x_{j_1,j_2},
	\end{equation}
	where $f_{\nu,j, i}\in \mathbb{F}_q$ and the index pair $(j_1,j_2)$ is defined in the $j$-th row of the index array $\mathbf{P}_{i-r}$ where $\nu \in [r]$. The $r$ index arrays $\mathbf{P}_1,\ldots,\mathbf{P}_r$ are defined as follows:
	\begin{equation*}
	\hspace{-2.9cm}
	\mathbf{P}_1=
		\begin{bmatrix}
			(1, 1) & (1, 2) & \ldots & (1, k)\\
			(2, 1) & (2, 2) & \ldots & (2, k)\\
			\vdotswithin{1} & \vdotswithin{\alpha_n} & \ddots & \vdotswithin{{\alpha_n}^{k-1}}\\
			(\alpha, 1) & (\alpha, 2) & \ldots & (\alpha, k)\\
		\end{bmatrix},
	\end{equation*}
	$$\ \ \ \ \ \ \ \ \ \ \ \ \ \ \ \ \ \ \ \ \ \ \ \ \ \ \ \ \ \ \ \ \ \ \ \ \ \ \ \ \ \overbrace{\ \ \ \ \ \ \ \ \ \ \ \ \ \ \ \ \ \ \ \ }^{\lceil \frac{k}{r} \rceil}$$
	\begin{equation*}
	\mathbf{P}_i=
	\begin{bmatrix}
	(1, 1) & (1, 2) & \ldots & (1, k) &  (?, ?) & \ldots & (?, ?) \\
	(2, 1) & (2, 2) & \ldots & (2, k) & (?, ?) & \ldots & (?, ?) \\
	\vdotswithin{1} & \vdotswithin{\alpha_n} & \ddots & \vdotswithin{{\alpha_n}^{k-1}}\\
	(\alpha, 1) & (\alpha, 2) & \ldots & (\alpha, k) & (?, ?) & \ldots & (?, ?) \\
	\end{bmatrix}.
	\end{equation*}
	where the values of the indexes $(?, ?)$ are determined by a scheduling algorithm that guarantees the code is MDS, i.e., the entire information $\mathbf{x}$ can be recovered from any $k$ out of the $n$ vectors $\mathbf{c}_i$. $\blacksquare$
\end{definition}

Algorithm \ref{Alg:Constr} gives a high level description of one scheduling algorithm for Def. \ref{HashTagCodes}. An interested reader is referred to \cite{7463553,8025778} for more details.

\begin{algorithm}
	\small
	\caption{High level description of an algorithm for generating HTEC for an arbitrary sub-packetization level
		\newline
		\textbf{Input:} $n, k, \alpha$;
		\newline
		\textbf{Output:} Index arrays $\mathbf{P}_1, \ldots, \mathbf{P}_r$.}
	\label{Alg:Constr}
	\begin{algorithmic}[1]
		\State{\textbf{Initialization:} $\mathbf{P}_1, \ldots, \mathbf{P}_r$ are initialized as index arrays $\mathbf{P} = ((i,j))_{\alpha \times k}$;}
		\State{Append $\Bigl\lceil \frac{k}{r}\Bigr\rceil$ columns to $\mathbf{P}_2, \ldots, \mathbf{P}_r$ all initialized to $(0, 0)$;}
		\State \# Phase 1
		\State Set the granulation level $run \leftarrow \Bigl\lceil \frac{\alpha}{r}\Bigr\rceil$ and $step \leftarrow 0$;
		\Repeat
		\State \parbox[t]{8cm}{Replace $(0, 0)$ pairs with indexes $(i, j)$ such that both Condition 1 and Condition 2 are satisfied;  \vspace{0.1cm}}
		\State \parbox[t]{8cm}{Decrease the granulation level $run$ by a factor $r$ and $step \leftarrow \Bigl\lceil \frac{\alpha}{r}\Bigr\rceil-run$;}
		\Until{The granulation level $run > 1$}
		\State \# Phase 2
		\State If there are still $(0, 0)$ and unscheduled elements from the systematic nodes, choose $(i, j)$ such that only Condition 2 is satisfied;
		\State Return the index arrays $\mathbf{P}_1, \ldots, \mathbf{P}_r$.
	\end{algorithmic}
\end{algorithm}

\begin{example}\label{Ex:HashTag0906alfa09}
	The linear expressions for the parity parts for a $(9, 6)$ HashTag code with $\alpha=9$ are given here. The way how we obtain them is explained in Section 4.1 in \cite{8025778}. We give one set of coefficients $f_{\nu,j, i}$ for equation (\ref{LinEquations}) from the finite field $\mathbb{F}_{32}$ with irreducible polynomial $x^5+x^3+1$. This code achieves the lower bound of repair bandwidth in (\ref{optimalMSR}), i.e., the repair bandwidth is $\gamma = \frac{8}{3} = 2.67$ for repair of any systematic node.
	Due to the big size $54 \times 81$, the systematic generator matrix of the associated scalar code is presented graphically in Fig. \ref{G0906alpha09} instead of presenting it numerically.

	{\footnotesize 
		\begin{tabular}{llll}
			\renewcommand{\arraystretch}{0.9}
			\hspace{-1.0cm}\begin{tabular}{l@{}l@{}l@{}l@{}l@{}l@{}l@{}l@{}l@{}l@{}l@{}l}
				$c_{1,7}=$ & $\textbf{\ 7}x_{1,1}$ &+&$\textbf{10}x_{1,2}$ &+& $\textbf{18}x_{1,3}$ &+& $\textbf{11}x_{1,4}$ &+& $\textbf{17}x_{1,5}$ &+& $\textbf{\ 6}x_{1,6}$\\
				$c_{2,7}=$ & $\textbf{26}x_{2,1}$ &+& $\textbf{17}x_{2,2}$ &+& $\textbf{25}x_{2,3}$ &+& $\textbf{27}x_{2,4}$ &+& $\textbf{31}x_{2,5}$ &+& $\textbf{\ 4}x_{2,6}$\\
				$c_{3,7}=$ & $\textbf{22}x_{3,1}$ &+& $\textbf{12}x_{3,2}$ &+& $\textbf{27}x_{3,3}$ &+& $\textbf{31}x_{3,4}$ &+& $\textbf{31}x_{3,5}$ &+& $\textbf{23}x_{3,6}$\\
				$c_{4,7}=$ & $\textbf{17}x_{4,1}$ &+& $\textbf{\ 9}x_{4,2}$ &+& $\textbf{14}x_{4,3}$ &+& $\textbf{\ 4}x_{4,4}$ &+& $\textbf{21}x_{4,5}$ &+& $\textbf{25}x_{4,6}$\\
				$c_{5,7}=$ & $\textbf{20}x_{5,1}$ &+& $\textbf{\ 5}x_{5,2}$ &+& $\textbf{\ 5}x_{5,3}$ &+& $\textbf{13}x_{5,4}$ &+& $\textbf{11}x_{5,5}$ &+& $\textbf{16}x_{5,6}$\\
				$c_{6,7}=$ & $\textbf{25}x_{6,1}$ &+& $\textbf{16}x_{6,2}$ &+& $\textbf{30}x_{6,3}$ &+& $\textbf{28}x_{6,4}$ &+& $\textbf{10}x_{6,5}$ &+& $\textbf{24}x_{6,6}$\\
				$c_{7,7}=$ & $\textbf{20}x_{7,1}$ &+& $\textbf{\ 8}x_{7,2}$ &+& $\textbf{21}x_{7,3}$ &+& $\textbf{\ 9}x_{7,4}$ &+& $\textbf{\ 3}x_{7,5}$ &+& $\textbf{25}x_{7,6}$\\
				$c_{8,7}=$ & $\textbf{23}x_{8,1}$ &+& $\textbf{\ 4}x_{8,2}$ &+& $\textbf{12}x_{8,3}$ &+& $\textbf{16}x_{8,4}$ &+& $\textbf{\ 8}x_{8,5}$ &+& $\textbf{17}x_{8,6}$\\
				$c_{9,7}=$ & $\textbf{\ 2}x_{9,1}$ &+& $\textbf{21}x_{9,2}$ &+& $\textbf{\ 8}x_{9,3}$ &+& $\textbf{16}x_{9,4}$ &+& $\textbf{\ 7}x_{9,5}$ &+& $\textbf{25}x_{9,6}$\\
			\end{tabular} & \\
			\ 
			\\
			\renewcommand{\arraystretch}{0.9}
			\centering
			\hspace{-1.0cm}\begin{tabular}{l@{}l@{}l@{}l@{}l@{}l@{}l@{}l@{}l@{}l@{}l@{}l@{}l@{}l@{}l@{}l}
				$c_{1,8}=$ & $\textbf{\ 8}x_{1,1}$ &+&$\textbf{24}x_{1,2}$ &+& $\textbf{21}x_{1,3}$ &+& $\textbf{19}x_{1,4}$ &+& $\textbf{\ 6}x_{1,5}$ &+& $\textbf{20}x_{1,6}$ &+& $\textbf{\ 8} {x_{4,1}}$ &+& $\textbf{\ 6}{x_{2,4}}$\\		
				$c_{2,8}=$ & $\textbf{\ 3}x_{2,1}$ &+& $\textbf{12}x_{2,2}$ &+& $\textbf{\ 6}x_{2,3}$ &+& $\textbf{\ 3}x_{2,4}$ &+& $\textbf{16}x_{2,5}$ &+& $\textbf{10}x_{2,6}$ &+& $\textbf{30}{x_{5,1}}$ &+& $\textbf{24}{x_{1,5}}$\\
				$c_{3,8}=$ & $\textbf{23}x_{3,1}$ &+& $\textbf{20}x_{3,2}$ &+& $\textbf{30}x_{3,3}$ &+& $\textbf{\ 7}x_{3,4}$ &+& $\textbf{16}x_{3,5}$ &+& $\textbf{10}x_{3,6}$ &+& $\textbf{21}{x_{6,1}}$ &+& $\textbf{27}{x_{1,6}}$\\
				$c_{4,8}=$ & $\textbf{14}x_{4,1}$ &+& $\textbf{\ 7}x_{4,2}$ &+& $\textbf{10}x_{4,3}$ &+& $\textbf{14}x_{4,4}$ &+& $\textbf{24}x_{4,5}$ &+& $\textbf{20}x_{4,6}$ &+& $\textbf{16}{x_{1,2}}$ &+& $\textbf{31}{x_{5,4}}$\\
				$c_{5,8}=$ & $\textbf{25}x_{5,1}$ &+& $\textbf{11}x_{5,2}$ &+& $\textbf{29}x_{5,3}$ &+& $\textbf{12}x_{5,4}$ &+& $\textbf{20}x_{5,5}$ &+& $\textbf{24}x_{5,6}$ &+& $\textbf{15}{x_{2,2}}$ &+& $\textbf{\ 6}{x_{4,5}}$\\
				$c_{6,8}=$ & $\textbf{17}x_{6,1}$ &+& $\textbf{27}x_{6,2}$ &+& $\textbf{\ 4}x_{6,3}$ &+& $\textbf{21}x_{6,4}$ &+& $\textbf{15}x_{6,5}$ &+& $\textbf{11}x_{6,6}$ &+& $\textbf{19}{x_{3,2}}$ &+& $\textbf{21}{x_{4,6}}$\\
				$c_{7,8}=$ & $\textbf{19}x_{7,1}$ &+& $\textbf{23}x_{7,2}$ &+& $\textbf{16}x_{7,3}$ &+& $\textbf{\ 4}x_{7,4}$ &+& $\textbf{14}x_{7,5}$ &+& $\textbf{16}x_{7,6}$ &+& $\textbf{\ 9}{x_{1,3}}$ &+& $\textbf{\ 8}{x_{8,4}}$\\
				$c_{8,8}=$ & $\textbf{\ 5}x_{8,1}$ &+& $\textbf{26}x_{8,2}$ &+& $\textbf{22}x_{8,3}$ &+& $\textbf{30}x_{8,4}$ &+& $\textbf{22}x_{8,5}$ &+& $\textbf{21}x_{8,6}$ &+& $\textbf{24}{x_{2,3}}$ &+& $\textbf{26}{x_{7,5}}$\\
				$c_{9,8}=$ & $\textbf{10}x_{9,1}$ &+& $\textbf{\ 8}x_{9,2}$ &+& $\textbf{10}x_{9,3}$ &+& $\textbf{27}x_{9,4}$ &+& $\textbf{28}x_{9,5}$ &+& $\textbf{20}x_{9,6}$ &+& $\textbf{16}{x_{3,3}}$ &+& $\textbf{\ 4}{x_{7,6}}$\\
			\end{tabular} & \\
			\ 
			\\
			\renewcommand{\arraystretch}{0.9}
			\centering
			\hspace{-1.0cm}\begin{tabular}{l@{}l@{}l@{}l@{}l@{}l@{}l@{}l@{}l@{}l@{}l@{}l@{}l@{}l@{}l@{}l}
				$c_{1,9}=$ & $\textbf{20}x_{1,1}$ &+& $\textbf{20}x_{1,2}$ &+& $\textbf{30}x_{1,3}$ &+& $\textbf{17}x_{1,4}$ &+& $\textbf{12}x_{1,5}$ &+& $\textbf{27}x_{1,6}$ &+& $\textbf{28}{x_{7,1}}$ &+& $\textbf{\ 9}{x_{3,4}}$\\
				$c_{2,9}=$ & $\textbf{18}x_{2,1}$ &+& $\textbf{10}x_{2,2}$ &+& $\textbf{20}x_{2,3}$ &+& $\textbf{21}x_{2,4}$ &+& $\textbf{13}x_{2,5}$ &+& $\textbf{\ 7}x_{2,6}$ &+& $\textbf{\ 2}{x_{8,1}}$ &+& $\textbf{\ 6}{x_{3,5}}$\\
				$c_{3,9}=$ & $\textbf{31}x_{3,1}$ &+& $\textbf{25}x_{3,2}$ &+& $\textbf{12}x_{3,3}$ &+& $\textbf{18}x_{3,4}$ &+& $\textbf{15}x_{3,5}$ &+& $\textbf{24}x_{3,6}$ &+& $\textbf{31}{x_{9,1}}$ &+& $\textbf{28}{x_{2,6}}$\\				
				$c_{4,9}=$ & $\textbf{\ 6}x_{4,1}$ &+& $\textbf{16}x_{4,2}$ &+& $\textbf{26}x_{4,3}$ &+& $\textbf{\ 4}x_{4,4}$ &+& $\textbf{21}x_{4,5}$ &+& $\textbf{27}x_{4,6}$ &+& $\textbf{26}{x_{7,2}}$ &+& $\textbf{\ 8}{x_{6,4}}$\\
				$c_{5,9}=$ & $\textbf{\ 7}x_{5,1}$ &+& $\textbf{\ 6}x_{5,2}$ &+& $\textbf{26}x_{5,3}$ &+& $\textbf{\ 6}x_{5,4}$ &+& $\textbf{15}x_{5,5}$ &+& $\textbf{16}x_{5,6}$ &+& $\textbf{28}{x_{8,2}}$ &+& $\textbf{\ 4}{x_{6,5}}$\\
				$c_{6,9}=$ & $\textbf{20}x_{6,1}$ &+& $\textbf{20}x_{6,2}$ &+& $\textbf{12}x_{6,3}$ &+& $\textbf{20}x_{6,4}$ &+& $\textbf{18}x_{6,5}$ &+& $\textbf{26}x_{6,6}$ &+& $\textbf{19}{x_{9,2}}$ &+& $\textbf{30}{x_{5,6}}$\\
				$c_{7,9}=$ & $\textbf{26}x_{7,1}$ &+& $\textbf{\ 2}x_{7,2}$ &+& $\textbf{\ 6}x_{7,3}$ &+& $\textbf{20}x_{7,4}$ &+& $\textbf{17}x_{7,5}$ &+& $\textbf{23}x_{7,6}$ &+& $\textbf{\ 8} {x_{4,3}}$ &+& $\textbf{31}{x_{9,4}}$\\
				$c_{8,9}=$ & $\textbf{20}x_{8,1}$ &+& $\textbf{15}x_{8,2}$ &+& $\textbf{13}x_{8,3}$ &+& $\textbf{20}x_{8,4}$ &+& $\textbf{10}x_{8,5}$ &+& $\textbf{24}x_{8,6}$ &+& $\textbf{31}{x_{5,3}}$ &+& $\textbf{\ 9}{x_{9,5}}$\\
				$c_{9,9}=$ & $\textbf{\ 6}x_{9,1}$ &+& $\textbf{\ 2}x_{9,2}$ &+& $\textbf{31}x_{9,3}$ &+& $\textbf{12}x_{9,4}$ &+& $\textbf{16}x_{9,5}$ &+& $\textbf{30}x_{9,6}$ &+& $\textbf{20}{x_{6,3}}$ &+& $\textbf{13}{x_{8,6}}$\\
			\end{tabular} & \\ \ \\
		\end{tabular}
	}
\end{example}

\begin{figure}
	\begin{minipage}[b]{0.5\linewidth}
		\centering
		\includegraphics[width=8.9cm,height=5cm]{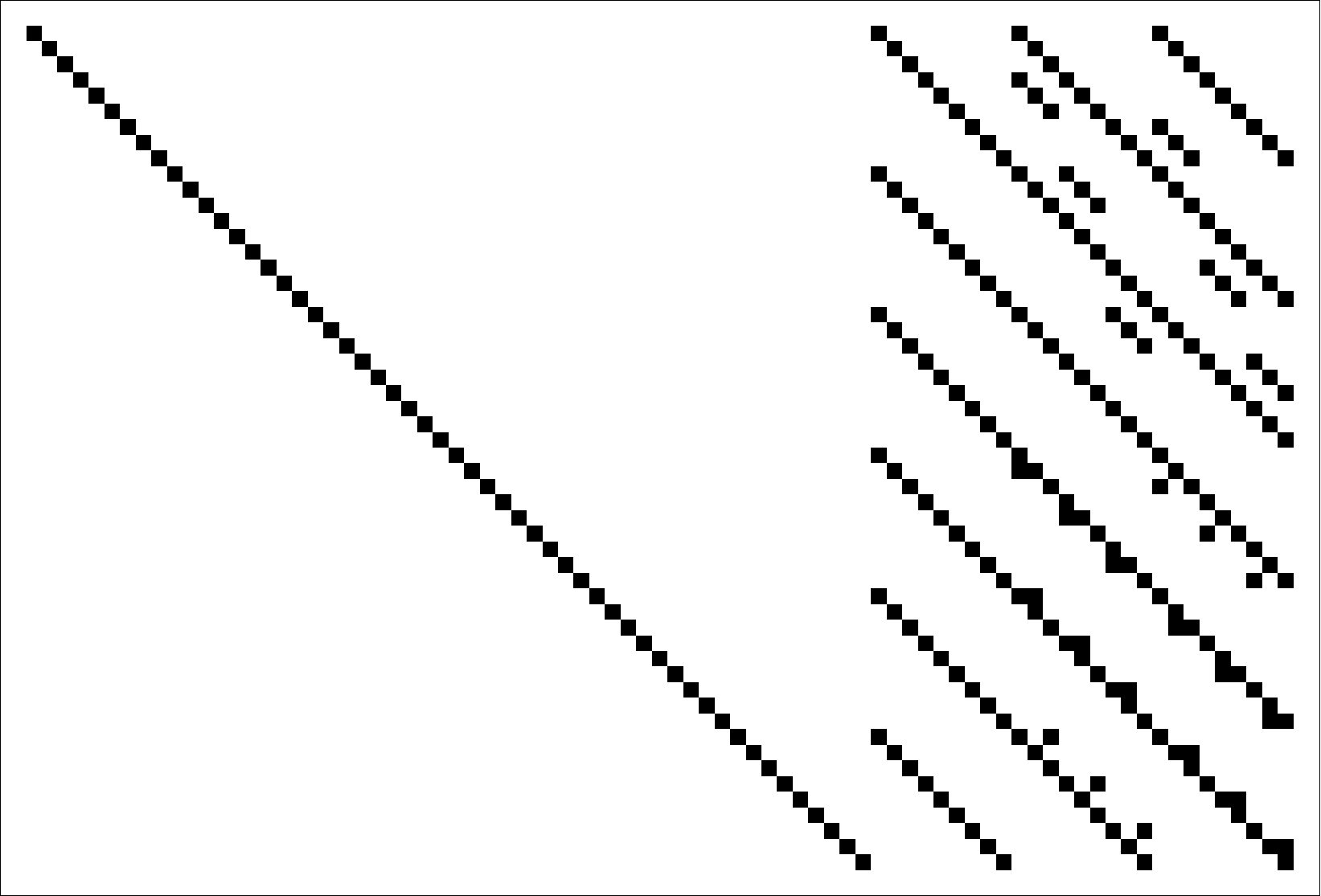}
	\end{minipage}
	\caption{A systematic generator matrix of the associated scalar code. Here the black squares on the main diagonal represent the value 1, but the black squares in the parity parts represent the non-zero values in $\mathbb{F}_{32}$. Note that if the parity matrix is partitioned in $9\times9$ square submatrices, it has the same form as in equation (\ref{parityPart}).}
	\label{G0906alpha09}
\end{figure}

We adapt the parity-splitting code construction for designing codes with local regeneration described in \cite{6846301} for the specifics of HashTag codes. The construction is described in Algorithm \ref{Alg:HashTagLRCConstruction}. For simplifying the description, we take some of the parameters to have specific relations, although it is possible to define a similar construction with general values of the parameters. Namely, we take that $r | k$ and $r | \alpha$. We also take that the parameters for the information locality $(l,\delta)$ are such that $l | k$ and $\delta \leq r$.
\begin{algorithm}
	\caption{Locally Repairable HashTag Codes
		\newline
		\textbf{Input:} A $(n, k)$ HashTag MDS code with a sub-packetization level $\alpha$ with the associated linear parity equations (\ref{LinEquations}), i.e. with the associated systematic generator matrix $\mathbf{G}$. The MDS code can be, but it does not necessarily have to be a MSR code.
		\newline
		\textbf{Input:} The information locality $(l,\delta)$
		\newline
		\textbf{Output:} A generator matrix $\mathbf{G}'$ with information locality $(l,\delta)$
	}
	 \begin{algorithmic}[1]
	 	\State Split $k$ systematic nodes into $l$ disjunctive subsets $S_i, i\in [l]$, where every set has $\frac{k}{l}$ nodes. While this splitting can be arbitrary, take the canonical splitting where $S_1 = \{1,\ldots, \frac{k}{l}\}$, $S_2 = \{\frac{k}{l}+1,\ldots, \frac{2 k}{l}\}$, $\ldots$, $S_{l} = \{\frac{(l-1)k}{l}+1,\ldots, k\}$.  
	 	\State Split each of the $\alpha$ linear equations for the first $\delta - 1$ parity expressions (\ref{LinEquations}) into $l$ sub-summands where the variables in each equation correspond to the elements from the disjunctive subsets.
		\State Associate the obtained $\alpha \times l \times (\delta - 1)$ sub-summands to $l \times (\delta - 1)$ new local parity nodes.
		\State Rename the remaining $r - \delta + 1$ parity nodes that were not split in Step 1 - Step 3 as new global parity nodes.
		\State Obtain a new systematic generator matrix $\mathbf{G}'$ from the local and global parity nodes.
		\State Return $\mathbf{G}'$ as a generator matrix of a $[n,K=k\alpha,d_{\min},\alpha,k]$  vector code with information locality $(l,\delta)$.
	 \end{algorithmic}
 \label{Alg:HashTagLRCConstruction}	
\end{algorithm}

\begin{figure}
	\begin{minipage}[b]{0.5\linewidth}
		\centering
		\includegraphics[width=8.5cm,height=4cm]{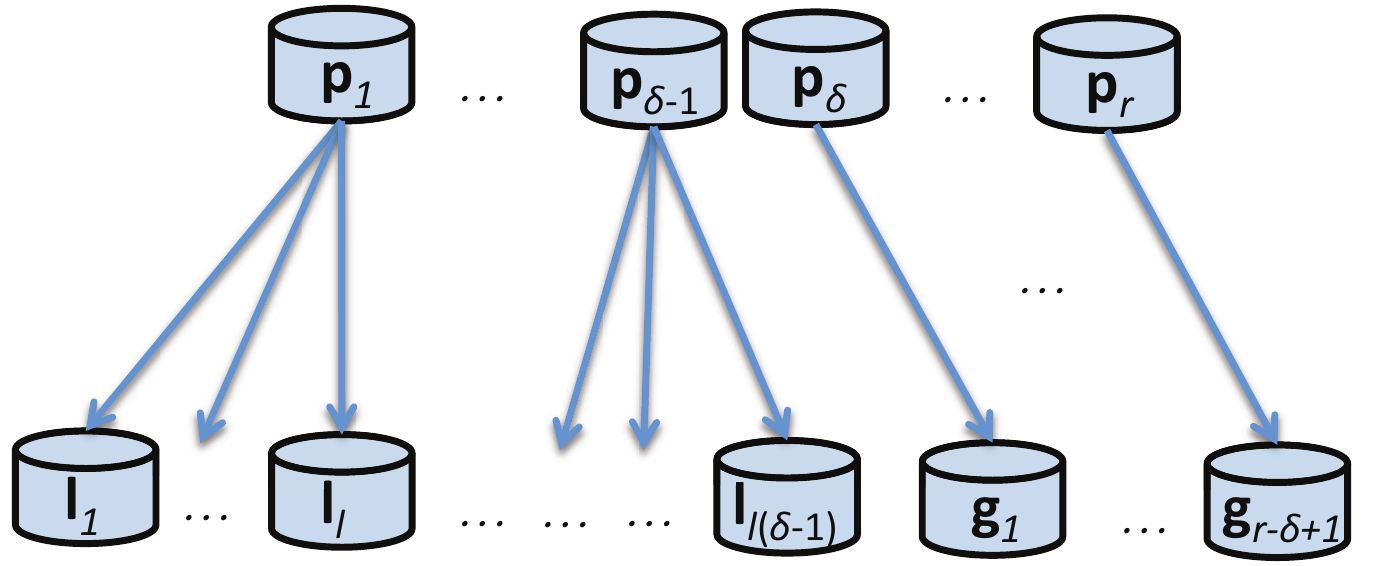}
	\end{minipage}
	\caption{There are $r$ parity nodes from a systematic $(n, k)$ MDS code with a sub-packetization level $\alpha$. The parity splitting technique generates $l$ local parity nodes from every parity node $\mathbf{p}_1,\ldots, \mathbf{p}_{\delta - 1}$ and renames the parity nodes $\mathbf{p}_{\delta},\ldots, \mathbf{p}_{r}$ as global parity nodes $\mathbf{g}_{1},\ldots, \mathbf{g}_{r-\delta+1}$. }	
	\label{framework}
\end{figure}

A graphical presentation of the parity-splitting procedure is given in Fig. \ref{framework}.

\begin{theorem}\label{Thm:LocalMSRHashTag}
	If the used $(n, k)$ MDS HashTag code in Algorithm \ref{Alg:HashTagLRCConstruction} is MSR, then the obtained $[n,K=k\alpha,d_{\min},\alpha,k]$ code with information locality $(l,\delta)$ is a MSR-Local code, where  
	\begin{equation}
	d_{\min} = n-k+1-\left( \frac{k}{l} - 1 \right)(\delta - 1) .
	\label{MinDistanceHashTagLocal}
	\end{equation}
\end{theorem}
\begin{proof} 
	{\rm Since in Algorithm \ref{Alg:HashTagLRCConstruction} we take that $r | k$ and $r | \alpha$, it means that the scalar dimension of the code is $K=m l \alpha$ for some integer $m$. Then the proof continues basically as a technical adaptation of the proof of Theorem 5.5 that Kamath et al. gave for the pyramid-like MSR-Local codes constructed with the parity-splitting strategy in \cite{6846301}.}
\end{proof}
Note that if $\alpha < r^{\frac{k}{r}}$, then HashTag codes are sub-optimal in terms of the repair bandwidth. Consequently, the produced codes with Algorithm \ref{Alg:HashTagLRCConstruction} are locally repairable, but they are not MSR-Local codes.

\begin{example}\label{Ex:HashTag0906alfa09-Split}
Let us split the MSR code given in Example \ref{Ex:HashTag0906alfa09} into a code with local regeneration and with information locality $(l=2,\delta=2)$. In Step 1 we split 6 systematic nodes $\{\mathbf{c}_{1},\ldots, \mathbf{c}_{6} \}$ into $l=2$ disjunctive subsets $S_1 = \{\mathbf{c}_{1}, \mathbf{c}_{2}, \mathbf{c}_{3} \}$ and $S_2 = \{\mathbf{c}_{4}, \mathbf{c}_{5}, \mathbf{c}_{6} \}$. According to Step 2 of Algorithm \ref{Alg:HashTagLRCConstruction}, the first global parity $\mathbf{c}_{7}$ in Example \ref{Ex:HashTag0906alfa09} is split into two local parities $\mathbf{l}_{1}=(l_{1,1}, \ldots,l_{9,1})^T$ and $\mathbf{l}_{2}=(l_{1,2},\ldots,l_{9,2})^T$ as follows:
\\
\\
	{\small
		\begin{tabular}{llll}
			\hspace{-0.5cm}
			\begin{tabular}{l@{}l@{}l@{}l@{}l@{}l@{}l@{}l@{}l@{}l@{}l@{}l}
				$l_{1,1}=$ & $\textbf{\ 7}x_{1,1}$ &+&$\textbf{10}x_{1,2}$ &+& $\textbf{18}x_{1,3}$\\
				$l_{2,1}=$ & $\textbf{26}x_{2,1}$ &+& $\textbf{17}x_{2,2}$ &+& $\textbf{25}x_{2,3}$\\
				$l_{3,1}=$ & $\textbf{22}x_{3,1}$ &+& $\textbf{12}x_{3,2}$ &+& $\textbf{27}x_{3,3}$\\
				$l_{4,1}=$ & $\textbf{17}x_{4,1}$ &+& $\textbf{\ 9}x_{4,2}$ &+& $\textbf{14}x_{4,3}$\\
				$l_{5,1}=$ & $\textbf{20}x_{5,1}$ &+& $\textbf{\ 5}x_{5,2}$ &+& $\textbf{\ 5}x_{5,3}$\\
				$l_{6,1}=$ & $\textbf{25}x_{6,1}$ &+& $\textbf{16}x_{6,2}$ &+& $\textbf{30}x_{6,3}$\\
				$l_{7,1}=$ & $\textbf{20}x_{7,1}$ &+& $\textbf{\ 8}x_{7,2}$ &+& $\textbf{21}x_{7,3}$\\
				$l_{8,1}=$ & $\textbf{23}x_{8,1}$ &+& $\textbf{\ 4}x_{8,2}$ &+& $\textbf{12}x_{8,3}$\\
				$l_{9,1}=$ & $\textbf{\ 2}x_{9,1}$ &+& $\textbf{21}x_{9,2}$ &+& $\textbf{\ 8}x_{9,3}$\\
			\end{tabular} & \ \ & 
			\begin{tabular}{l@{}l@{}l@{}l@{}l@{}l@{}l@{}l@{}l@{}l@{}l@{}l}
				$l_{1,2}=$ & $\textbf{11}x_{1,4}$ &+& $\textbf{17}x_{1,5}$ &+& $\textbf{\ 6}x_{1,6}$\\
				$l_{2,2}=$ & $\textbf{27}x_{2,4}$ &+& $\textbf{31}x_{2,5}$ &+& $\textbf{\ 4}x_{2,6}$\\
				$l_{3,2}=$ & $\textbf{31}x_{3,4}$ &+& $\textbf{31}x_{3,5}$ &+& $\textbf{23}x_{3,6}$\\
				$l_{4,2}=$ & $\textbf{\ 4}x_{4,4}$ &+& $\textbf{21}x_{4,5}$ &+& $\textbf{25}x_{4,6}$\\
				$l_{5,2}=$ & $\textbf{13}x_{5,4}$ &+& $\textbf{11}x_{5,5}$ &+& $\textbf{16}x_{5,6}$\\
				$l_{6,2}=$ & $\textbf{28}x_{6,4}$ &+& $\textbf{10}x_{6,5}$ &+& $\textbf{24}x_{6,6}$\\
				$l_{7,2}=$ & $\textbf{\ 9}x_{7,4}$ &+& $\textbf{\ 3}x_{7,5}$ &+& $\textbf{25}x_{7,6}$\\
				$l_{8,2}=$ & $\textbf{16}x_{8,4}$ &+& $\textbf{\ 8}x_{8,5}$ &+& $\textbf{17}x_{8,6}$\\
				$l_{9,2}=$ & $\textbf{16}x_{9,4}$ &+& $\textbf{\ 7}x_{9,5}$ &+& $\textbf{25}x_{9,6}$\\
			\end{tabular}
		\end{tabular}		
	}
\\
	The remaining two global parities are kept as they are given in Example \ref{Ex:HashTag0906alfa09}, they are only renamed as $\mathbf{g}_{1}=(c_{1,8}, c_{2,8},\ldots,c_{9,8})^T$ and $\mathbf{g}_{2}=(c_{1,9}, c_{2,9},\ldots,c_{9,9})^T$. The overall code is a $(10, 6)$ code or with the terminology from \cite{conf/usenix/HuangSXOCG0Y12} it is a $(6, 2, 2)$ code. $\blacksquare$
\end{example}

\begin{example}\label{Ex:HashTag0906alfa09-Split3}
	Let us split the same MSR code now with parameters $(l=3,\delta=2)$. In Step 1 we split 6 systematic nodes $\{\mathbf{c}_{1},\ldots, \mathbf{c}_{6} \}$ into $l=3$ disjunctive subsets $S_1 = \{\mathbf{c}_{1}, \mathbf{c}_{2}\}$, $S_2 = \{\mathbf{c}_{3}, \mathbf{c}_{4}\}$ and $S_3 = \{\mathbf{c}_{5}, \mathbf{c}_{6}\}$. In Step 2 of Algorithm \ref{Alg:HashTagLRCConstruction}, the first global parity $\mathbf{c}_{7}$ is split into three local parities: $\mathbf{l}_{1}=(l_{1,1}, \ldots,l_{9,1})^T$,  $\mathbf{l}_{2}=(l_{1,2},\ldots,l_{9,2})^T$ and 
	$\mathbf{l}_{3}=(l_{1,3},\ldots,l_{9,3})^T$	as follows:
	\\
	\\
	{\small
	\begin{tabular}{l@{\ }l@{}l@{\ }l@{}l}
		\hspace{-0.5cm}
		\begin{tabular}{l@{}l@{}l@{}l@{}l@{}l@{}l@{}l@{}l@{}l@{}l@{}l}
			$l_{1,1}=$ & $\textbf{\ 7}x_{1,1}$ &+&$\textbf{10}x_{1,2}$    \\
			$l_{2,1}=$ & $\textbf{26}x_{2,1}$  &+& $\textbf{17}x_{2,2}$   \\
			$l_{3,1}=$ & $\textbf{22}x_{3,1}$  &+& $\textbf{12}x_{3,2}$   \\
			$l_{4,1}=$ & $\textbf{17}x_{4,1}$  &+& $\textbf{\ 9}x_{4,2}$  \\
			$l_{5,1}=$ & $\textbf{20}x_{5,1}$  &+& $\textbf{\ 5}x_{5,2}$  \\
			$l_{6,1}=$ & $\textbf{25}x_{6,1}$  &+& $\textbf{16}x_{6,2}$   \\
			$l_{7,1}=$ & $\textbf{20}x_{7,1}$  &+& $\textbf{\ 8}x_{7,2}$  \\
			$l_{8,1}=$ & $\textbf{23}x_{8,1}$  &+& $\textbf{\ 4}x_{8,2}$  \\
			$l_{9,1}=$ & $\textbf{\ 2}x_{9,1}$ &+& $\textbf{21}x_{9,2}$   \\
		\end{tabular} && 
		\begin{tabular}{l@{}l@{}l@{}l@{}l@{}l@{}l@{}l@{}l@{}l@{}l@{}l}
			$l_{1,2}=$ & $\textbf{18}x_{1,3}$  &+& $\textbf{11}x_{1,4}$ \\
			$l_{2,2}=$ & $\textbf{25}x_{2,3}$  &+& $\textbf{27}x_{2,4}$ \\
			$l_{3,2}=$ & $\textbf{27}x_{3,3}$  &+& $\textbf{31}x_{3,4}$ \\
			$l_{4,2}=$ & $\textbf{14}x_{4,3}$  &+& $\textbf{\ 4}x_{4,4}$\\
			$l_{5,2}=$ & $\textbf{\ 5}x_{5,3}$ &+& $\textbf{13}x_{5,4}$ \\
			$l_{6,2}=$ & $\textbf{30}x_{6,3}$  &+& $\textbf{28}x_{6,4}$ \\
			$l_{7,2}=$ & $\textbf{21}x_{7,3}$  &+& $\textbf{\ 9}x_{7,4}$\\
			$l_{8,2}=$ & $\textbf{12}x_{8,3}$  &+& $\textbf{16}x_{8,4}$ \\
			$l_{9,2}=$ & $\textbf{\ 8}x_{9,3}$ &+& $\textbf{16}x_{9,4}$ \\
		\end{tabular} &&
		\begin{tabular}{l@{}l@{}l@{}l@{}l@{}l@{}l@{}l@{}l@{}l@{}l@{}l}
			$l_{1,3}=$ & $\textbf{17}x_{1,5}$  &+& $\textbf{\ 6}x_{1,6}$\\
			$l_{2,3}=$ & $\textbf{31}x_{2,5}$  &+& $\textbf{\ 4}x_{2,6}$\\
			$l_{3,3}=$ & $\textbf{31}x_{3,5}$  &+& $\textbf{23}x_{3,6}$\\
			$l_{4,3}=$ & $\textbf{21}x_{4,5}$  &+& $\textbf{25}x_{4,6}$\\
			$l_{5,3}=$ & $\textbf{11}x_{5,5}$  &+& $\textbf{16}x_{5,6}$\\
			$l_{6,3}=$ & $\textbf{10}x_{6,5}$  &+& $\textbf{24}x_{6,6}$\\
			$l_{7,3}=$ & $\textbf{\ 3}x_{7,5}$ &+& $\textbf{25}x_{7,6}$\\
			$l_{8,3}=$ & $\textbf{\ 8}x_{8,5}$ &+& $\textbf{17}x_{8,6}$\\
			$l_{9,3}=$ & $\textbf{\ 7}x_{9,5}$ &+& $\textbf{25}x_{9,6}$\\
		\end{tabular}
	\end{tabular}		
	}
\\
	The remaining two global parities are kept as they are given in Example \ref{Ex:HashTag0906alfa09}, but they are just renamed as $\mathbf{g}_{1}=(c_{1,8}, c_{2,8},\ldots,c_{9,8})^T$ and $\mathbf{g}_{2}=(c_{1,9}, c_{2,9},\ldots,c_{9,9})^T$. The overall code is a $(11, 6)$ code or with the terminology from \cite{conf/usenix/HuangSXOCG0Y12} it is a $(6, 3, 2)$ code. $\blacksquare$
\end{example}

There are two interesting aspects of Theorem \ref{Thm:LocalMSRHashTag} that should be emphasized: {\bf 1.} We give an explicit construction of an MSR-Local code (note that in \cite{6846301} the construction is existential), and {\bf 2.} Examples \ref{Ex:HashTag0906alfa09-Split} and \ref{Ex:HashTag0906alfa09-Split3} show that the size of the finite field can be slightly lower than the size proposed in \cite{6846301}. Namely, the MSR HashTag code used in our example is defined over $\mathbb{F}_{32}$, while the lower bound in \cite{6846301} suggests the field size to be bigger than $\binom{9}{6} = 84$. We consider this as a minor contribution and an indication that a deeper theoretical analysis can further lower the field size bound given in \cite{6846301}.

\section{Repair Duality}\label{repair}
\begin{theorem}\label{Thm:DoubleNature}
		Let $\mathcal{C}$ be a $(n, k)$ MSR HashTag code with $\gamma_{MSR}^{min}=\frac{M}{k}\frac{n-1}{n-k}$. Further, let $\mathcal{C}'$ be a $[n,K=k\alpha,d_{\min},\alpha,k]$ code with local regeneration and with information locality $(l,\delta)$ obtained by Algorithm \ref{Alg:HashTagLRCConstruction}. If we denote with $\gamma_{Local}^{min}$ the minimum repair bandwidth for single systematic node repair with $\mathcal{C}'$, then  
		\begin{equation}
		\gamma_{Local}^{min}=\min(\frac{M}{k}\frac{\frac{k}{l}+\delta - 2}{\delta - 1}, \frac{M}{k}\frac{n-1}{n-k}).
		\label{optimalBW}
		\end{equation}
\end{theorem}
\begin{proof}
	When repairing one systematic node, we can always treat local nodes as virtual global nodes from which they have been constructed by splitting. Then with the use of other global nodes we have a situation of repairing one systematic node in the original MSR code for which the repair bandwidth is $\frac{M}{k}\frac{n-1}{n-k}$. On the other hand, if we use the MSR-Local code, then we have the following situation. There are $\frac{k}{l}$ systematic nodes in the MSR-Local code, and the total length of the MSR-Local code is $\frac{k}{l} + \delta - 1$. The file size for the MSR-Local code is decreased by a factor $l$, i.e., it is $\frac{M}{l}$. If we apply the MSR repair bandwidth for these values we get:
	$$\frac{\frac{M}{l}}{\frac{k}{l}} \cdot \frac{\frac{k}{l}+(\delta - 1)-1}{\delta - 1} =  \frac{M}{k} \frac{\frac{k}{l}+\delta - 2}{\delta - 1}.$$
\end{proof}
		

Theorem \ref{Thm:DoubleNature} is one of the main contributions of this work: It emphasizes the \emph{repair duality} for repairing one systematic node: by the local and global parity nodes or only by the local parity nodes. We want to emphasize the practical importance of Theorem \ref{Thm:DoubleNature}. Namely, in practical implementations regardless of the theoretical value of $\gamma_{Local}^{min}$, the number of I/O operations and the access time for the contacted parts can be either crucial or insignificant. In those cases an intelligent repair strategy implemented in the distributed storage system can decide which repair procedure should be used: the one with global parity nodes or the one with the local parity nodes.

While the repair bandwidth in equation (\ref{optimalBW}) decreases by increasing the values of $\delta$ and $l$, it comes at the cost of decreasing the rate of the code for introducing the repair locality. We formalize this by the following Proposition.
\begin{proposition}\label{Prop:OutputCodeDimensions}
If the input code in Algorithm \ref{Alg:HashTagLRCConstruction} is $(n, k)$, then the output locally repairable code is \begin{equation}\label{eqn:OutputCodeDimensions}
    (n+ l\times (\delta - 1) -\delta + 1, k).
\end{equation}
\end{proposition} 
\begin{proof}
     For a given initial code $(n,k)$ the total number of parity nodes in the produced locally repairable codes is a sum of the parity nodes counted in Step 3 and Step 4 in Algorithm \ref{Alg:HashTagLRCConstruction}. The sum is $$l\times (\delta - 1) + r -\delta +1. $$ Thus the total number of nodes in the output locally repairable codes is $$k + l\times (\delta - 1) + r -\delta +1 =  n+ l\times (\delta - 1) -\delta + 1.$$
\end{proof}

The penalty paid by the decremented rate of the final locally repairable code as described in the expression (\ref{eqn:OutputCodeDimensions}) stays linear in $l$ if $\delta=2$, but increases by a multiplicative factor $l\times (\delta - 1)$ if $\delta \geq 3$. That explains the reasons why in practical implementations of locally repairable codes such as those in Windows Azure \cite{conf/usenix/HuangSXOCG0Y12} and the initial definition of locally repairable codes introduced in \cite{journals/tit/GopalanHSY12}, the value of $\delta$ is kept low, i.e., $\delta=2$.

\begin{figure*}[h!]
	\centering
	\subfloat[][]{\includegraphics[width=0.48\textwidth]{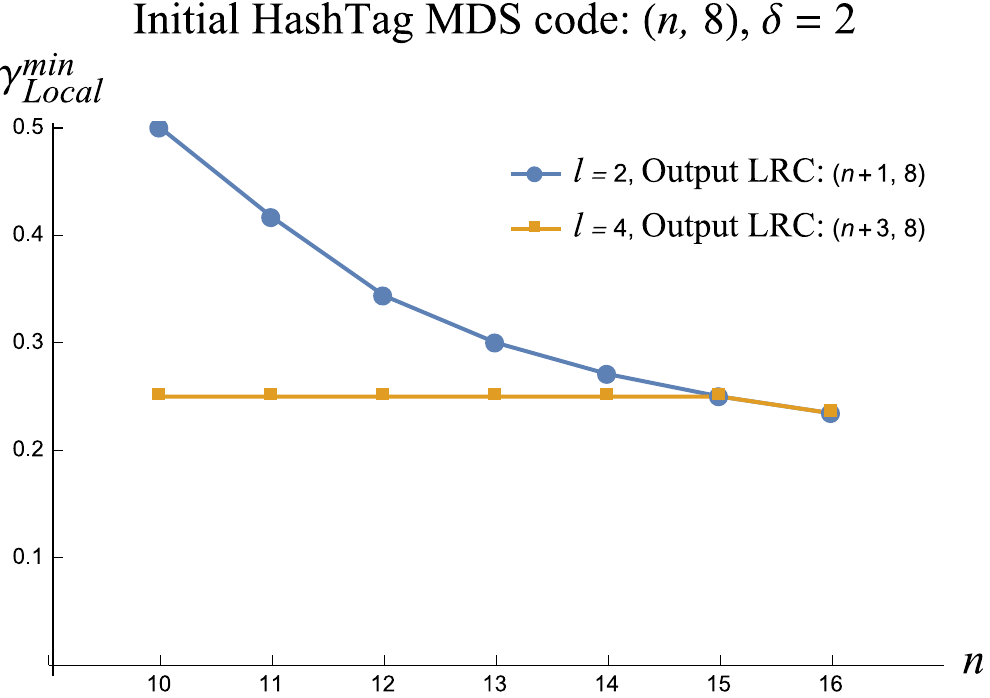}} 
	\hfill
	\subfloat[][]{\includegraphics[width=0.48\textwidth]{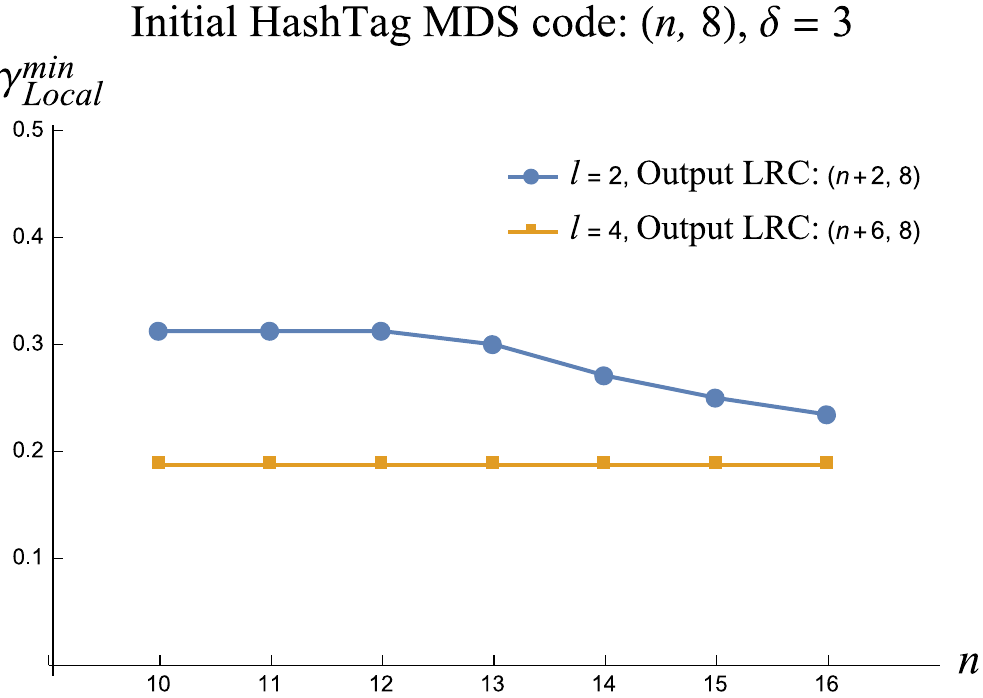}} 
	\caption{Repair bandwidth for one node for different codes where $k=8$. The initial HashTag codes that are input in Algorithm \ref{Alg:HashTagLRCConstruction} are $(n, 8)$ for different values of $n$. The left sub-figure (a) is for $\delta=2$ and the right sub-figure (b) is for $\delta=3$.}
	\label{fig:k8}
\end{figure*} 
\begin{figure*}
	\centering
	\subfloat[][]{\includegraphics[width=0.48\textwidth]{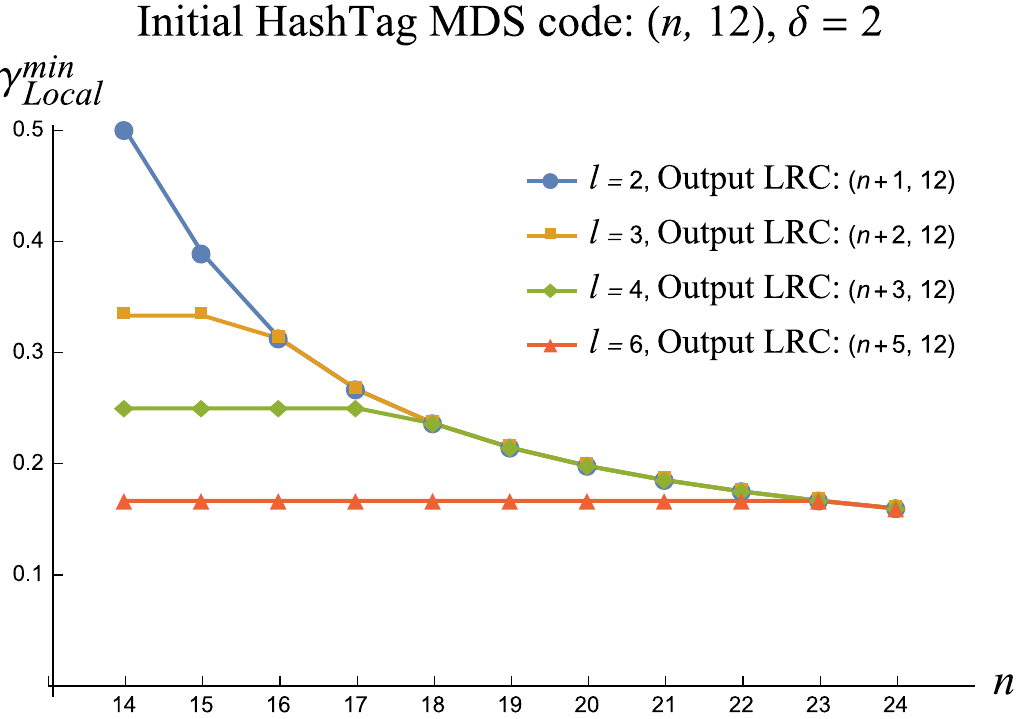}} 
	\hfill
	\subfloat[][]{\includegraphics[width=0.48\textwidth]{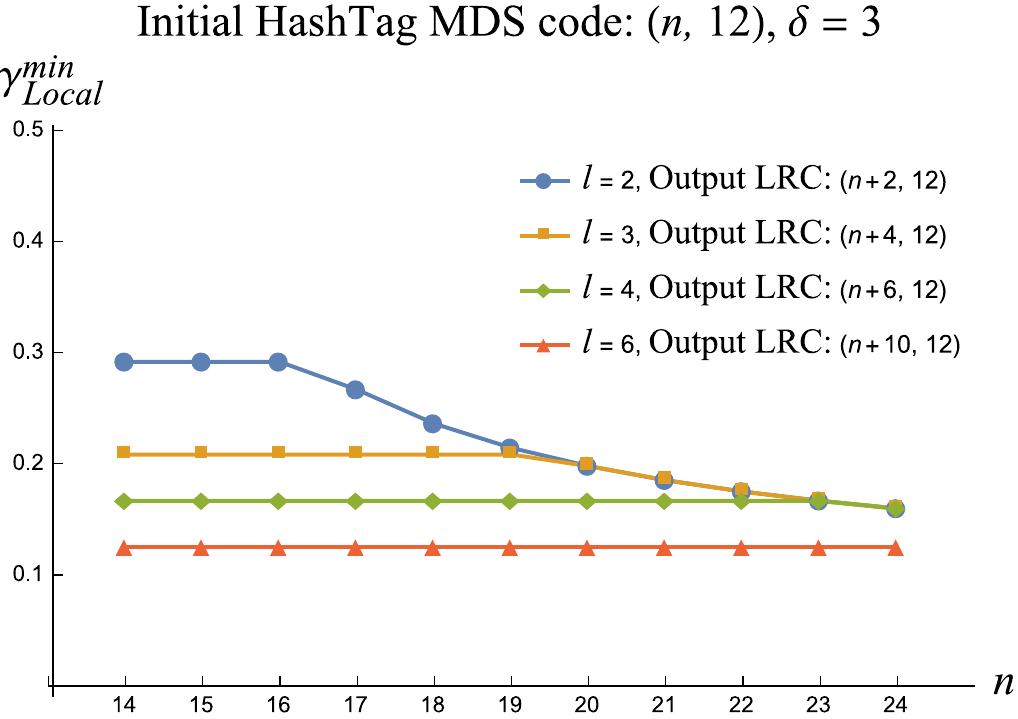}} 
	\caption{Repair bandwidth for one node for different codes where $k=12$. The initial HashTag codes that are input in Algorithm \ref{Alg:HashTagLRCConstruction} are $(n, 12)$ for different values of $n$. The left sub-figure (a) is for $\delta=2$ and the right sub-figure (b) is for $\delta=3$.}
	\label{fig:k12}
\end{figure*} 

To further illustrate different choices for introducing locality and the existence of repair duality for those codes, in Fig. \ref{fig:k8} and Fig. \ref{fig:k12} we plot the values for $\gamma_{Local}^{min}$ from the expression (\ref{optimalBW}) for two families of codes with $k=8$ and $k=12$. We took a normalized value $M=1$. For comparative reasons, on the left sub-figure we plot the values for $\delta=2$ and on the right sub-figure we plot the values for $\delta=3$. While in principle, the values of $l$ do not necessarily need to divide the value of $k$, in Algorithm \ref{Alg:HashTagLRCConstruction} it is very convenient if actually $l$ divides $k$. That is the reason why in Fig. \ref{fig:k8} that is produced for $(n,k)$ codes where $k=8$ we plot the values for codes obtained for $l=2$ and $l=4$. Similarly, the values for $l$ are $l=2, 3, 4$ and $6$ for the codes $(n,12)$ in Fig. \ref{fig:k12}.

The flat lines in Fig. \ref{fig:k8} and Fig. \ref{fig:k12} mean that the repair bandwidth in expression (\ref{optimalBW})  is achieved by using the local parity nodes, i.e., the value is  $\frac{M}{k}\frac{\frac{k}{l}+\delta - 2}{\delta - 1}$. Since this expression does not depend on $n$, its value is a constant for different values of $n$. However, in cases when the minimum in expression (\ref{optimalBW})  is achieved by using the global parity nodes, the repair bandwidth is $\frac{M}{k}\frac{n-1}{n-k}$ and the plot of the values $\gamma_{Local}^{min}$ has a decreasing shape for increasing values of $n$.

Another important aspect that is illustrated by Fig. \ref{fig:k8} and Fig. \ref{fig:k12} is the price that is paid for achieving a low repair bandwidth with locally repairable codes. For example let us take an initial HashTag code $(n, k) = (10, 8)$. A LRC code produced with information locality $(l=2,\delta=2)$ is a $(n+1, 8) = (11, 8)$ code and has a repair bandwidth of $0.5$. By using information locality $(l=4,\delta=2)$, the LRC code is a $(n+3, 8) = (13, 8)$ code, but the repair bandwidth drops down to $0.25$. The situation with $\delta=3$ decreases the repair bandwidth further, but worsens the code rate as well. In Fig. \ref{fig:k8}(b) an initial HashTag code $(n, k) = (10, 8)$ produces a LRC code $(n+2, 8) = (12, 8)$ with a repair bandwidth of $0.3125$ for $l=2$, and produces a LRC code $(n+6, 8) = (16, 8)$ with a repair bandwidth of just $0.1875$ for $l=4$.

We illustrate the benefits of having a choice how the repair is done (either by the local nodes or by the global nodes) by the following practical example.

\begin{example}\label{Ex:Duality}
	Let us consider the $(9, 6)$ MSR HashTag code given in Example \ref{Ex:HashTag0906alfa09} and its corresponding local variant from Algorithm \ref{Alg:HashTagLRCConstruction} with information locality $(l=2,\delta=2)$ given in Example \ref{Ex:HashTag0906alfa09-Split}. That means that the code with local regeneration has 6 systematic nodes, 2 local and 2 global parity nodes. 
	
	Let us analyze the number of reads when we recover one unavailable systematic node. If we recover with the local nodes, then we have to perform 3 sequential reads, reading the whole data in a contiguous manner from 3 nodes. If we repair the unavailable data with the help of both local and global parity nodes, it reduces to the case of recovery with a MSR code, where the number of sequential reads is between 8 and 24 (average 16 reads) but the amount of transferred data is equivalent to 2.67 nodes. 
	
	More concretely, let us assume that we want to recover the node $\mathbf{x}_1 = (x_{1,1}, x_{2,1},\ldots,x_{9,1})^T$. 
	\begin{enumerate}
		\item For a recovery only with the local parity $\mathbf{l}_{1}$, 3 sequential reads of $\mathbf{l}_{1}$, $\mathbf{x}_{2}$ and $\mathbf{x}_{3}$ are performed.
		\item For a recovery with the local and global parities:
		\begin{enumerate}
			\item First, read $l_{1,1}$,  $l_{2,1}$ and $l_{3,1}$ from $\mathbf{l}_{1}$, and $x_{1,2}$,  $x_{2,2}$ and $x_{3,2}$ from $\mathbf{x}_{2}$ and $x_{1,3}$, $x_{2,3}$ and $x_{3,3}$ from $\mathbf{x}_{3}$ to recover $x_{1,1}$,  $x_{2,1}$ and $x_{3,1}$.
			\item Additionally, read $x_{1,4}$,  $x_{2,4}$ and $x_{3,4}$ from $\mathbf{x}_{4}$ and $x_{1,5}$,  $x_{2,5}$ and $x_{3,5}$ from $\mathbf{x}_{5}$ and $x_{1,6}$, $x_{2,6}$ and $x_{3,6}$ from $\mathbf{x}_{6}$.
			\item Then, read $c_{1,8}$,  $c_{2,8}$ and $c_{3,8}$ from the global parity $\mathbf{g}_{1}$ to recover $x_{4,1}$,  $x_{5,1}$ and $x_{6,1}$.
			\item Finally, read $c_{1,9}$,  $c_{2,9}$ and $c_{3,9}$ from the global parity $\mathbf{g}_{2}$ to recover $x_{7,1}$,  $x_{8,1}$ and $x_{9,1}$.
		\end{enumerate}
	\end{enumerate}

	Now, let a small file of 54 KB be stored across 6 systematic, 2 local and 2 global parity nodes. The sub-packetization level is $\alpha=9$, thus every node stores 9 KB, sub-packetized in 9 parts, each of size 1 KB. If the access time for starting a read operation is approximately the same as transferring 9 KB, then repairing with local and global parity nodes is more expensive since we have to perform in average 12 reads, although the amount of transferred data is equivalent to 2.67 nodes.
	
	On the other hand, let us have a big file of 540 MB stored across 6 systematic nodes and 2 local and 2 global parity nodes. The sub-packetization level is again $\alpha=9$, thus every node stores 90 MB, sub-packetized in 9 parts, each of size 10 MB. The access time for starting a read operation is again approximately the same as transferring 9 KB, which is insignificant in comparison with the total amount of transferred data in the process of repairing of a node. In this case, it is better to repair a failed node with local and global parity nodes since it requires a transfer of 240 MB versus the repair just with local nodes that requires a transfer of 270 MB.
\end{example}

\section{HashTag LRC Codes With Efficient Recovery of a Global Parity Node} \label{global}
HashTag codes as well as different MSR codes \cite{7902203} and LRC codes \cite{journals/tit/GopalanHSY12,conf/infocom/OggierD11,6195703} have the property that the bandwidth for a recovery of a failed parity node (global parity nodes in the case of LRC) is equal as in Reed-Solomon codes. That means the recovery of a failed parity node is not optimal and requires a bandwidth of $k$ nodes. 

The problem of recovery of a parity node that is optimal and achieves the MSR bound was recently solved by Tian et al. in \cite{8006804}. Their approach is to take any $(n, k)$ MDS code, where $r = n - k$, and then to increase the sub-packetization level by a factor $r$. Thus, by producing a HashTag code with Algorithm \ref{Alg:Constr}, and applying first the technique  from \cite{8006804}, and then the parity-splitting technique defined in Algorithm \ref{Alg:HashTagLRCConstruction} we can construct HashTag LRC codes that can efficiently recover a global parity node. However, it comes at the cost of an increased sub-packetization level from $\alpha$ to $r \alpha$.

Here, for $\delta = 2$, i.e., for codes with information locality $(l, 2)$ we present another approach in Algorithm \ref{Alg:HashTagLRCWithEfficientGlobalParityRecovery} that does not increase the sub-packetization level $\alpha$ of the initial HashTag code.

\begin{algorithm}
	\caption{Locally Reparable HashTag Codes with Efficient Recovery of a Global Parity Node
		\newline
		\textbf{Input:} Number of data nodes $k$, information locality $(l, 2)$ and number of global nodes $g$.
		\newline
		\textbf{Output:} A generator matrix $\mathbf{G'}$ with information locality $(l,2)$, and sub-packetization level $\alpha = g$.
	}
	\begin{algorithmic}[1]
		\State With Algorithm \ref{Alg:Constr} produce a MDS HashTag code with $k$ data nodes, $r = g + l - 1$ parity nodes and $\alpha = g$ substripes.
		\State 	Get the corresponding index arrays $P_1,\ldots, P_r$.
		\State Set global parity nodes $\{\mathbf{g}_1,\ldots,\mathbf{g}_\alpha \}$ = $\{\mathbf{p}_2, \ldots, \mathbf{p}_r\}$.
		\State Set the substripes of the node $\mathbf{g}_i$ as $\mathbf{g}_{i}=(g_{1,i}, \ldots,g_{\alpha,i})^T$.
		\State Set the matrix of all global substripes 		
		$$
		\mathbf{G}_{\alpha \times \alpha}[g_{i,j}] =
		\begin{bmatrix}
		g_{1, 1} & g_{1, 2} & \ldots & g_{1, \alpha}\\
		g_{2, 1} & g_{2, 2} & \ldots & g_{2, \alpha}\\
		\vdotswithin{1} & \vdotswithin{\alpha_n} & \ddots & \vdotswithin{{\alpha_n}^{k-1}}\\
		g_{\alpha, 1} & g_{\alpha, 2} & \ldots & g_{\alpha, \alpha}\\
		\end{bmatrix}
		$$	
		\State Split $k$ systematic nodes into $l$ disjunctive subsets $S_i, i\in [l]$, where every set has $\frac{k}{l}$ nodes. While this splitting can be arbitrary, take the canonical splitting where $S_1 = \{1,\ldots, \frac{k}{l}\}$, $S_2 = \{\frac{k}{l}+1,\ldots, \frac{2 k}{l}\}$, $\ldots$, $S_{l} = \{\frac{(l-1)k}{l}+1,\ldots, k\}$.  
		\State Split each of the $\alpha$ linear equations for the first parity expression in (\ref{LinEquations}) into $l$ sub-summands where the variables in each equation correspond to the elements from the disjunctive subsets.
		\State Associate the obtained $\alpha \times l $ sub-summands to $l $ new local parity nodes.
		\State From $\mathbf{G}_{\alpha \times \alpha}[g_{i,j}]$ obtain a new systematic generator matrix $\mathbf{G'}_{\alpha \times \alpha}[g'_{i,j}] $ as follows: \begin{equation}\label{eq:NewGMatrix}
		g'_{i,j} = \left\{
		\begin{array}{cl}
		g_{i,j} & \text{if } i = j,\\
		f_{i, j, 1} g_{i,j} + f_{i, j, 2} g_{j,i} & \text{if } i < j, \\
		f_{i, j, 3} g_{i,j} + f_{i, j, 4} g_{j,i} & \text{if } i > j, \\
		\end{array} \right.
		\end{equation}
		where the coefficients $f_{i, j, 1}, \ldots, f_{i, j, 4} \in \mathbb{F}_q$ form $2\times 2$ non-singular matrices 
		$
		\begin{bmatrix}
		f_{i, j, 1} & f_{i, j, 2} \\
		f_{i, j, 3} & f_{i, j, 4} \\
		\end{bmatrix}
		$.
		\State Return $\mathbf{G'}$ as a generator matrix of a $[k+l+g, K=k g, d_{\min}, g, k]$ vector code with information locality $(l,2)$.
	\end{algorithmic}
	\label{Alg:HashTagLRCWithEfficientGlobalParityRecovery}	
\end{algorithm}

\begin{theorem}
	The bandwidth for repair of one global node produced by Algorithm \ref{Alg:HashTagLRCWithEfficientGlobalParityRecovery} is equal to the bandwidth for repair of one data node of a MDS HashTag produced by Algorithm \ref{Alg:Constr}.
\end{theorem}
\begin{proof}
	Without a loss of generality let us assume that the lost global parity node is $\mathbf{g}_{1}=(g'_{1,1}, \ldots,g'_{\alpha,1})^T$. From the relations (\ref{eq:NewGMatrix}) it follows that we can reconstruct $g'_{1,1}$ by reading the first $k$ substripes from all $k$ data nodes $x_{1,1}, \ldots, x_{\alpha,1} $ plus one element $x_{\mu, \nu}$ where the concrete values for $\mu$ and $\nu$ are obtained from the output of Algorithm \ref{Alg:Constr}. That is the same amount of bandwidth as with the recovery of the first substripe of a data node in the original HashTag. 
	
	For recovery of the substripe $g'_{2,1}$ we use equation (\ref{eq:NewGMatrix}) and we read the substripe $g'_{1,2}$ from the second global node. Since we already have read the values $x_{1,1}, \ldots, x_{\alpha,1} $, by using the coefficients $f_{1, 2, 1}$, $f_{1, 2, 2}$, $f_{1, 2, 3}$ and $f_{1, 2, 4}$ and possibly reading one extra stripe element $x_{\mu, \nu}$ where the concrete values for $\mu$ and $\nu$ are obtained from the output of Algorithm \ref{Alg:Constr} we can compute the values $g_{i,j}$ and  $g_{j,i}$, i.e., we compute the value $g'_{2,1} = f_{1, 2, 3} g_{1,2} + f_{1, 2, 4} g_{2,1}$. Total number of substripes read in this step is 2 (the same as in the original HashTag algorithm when recovering the second substripe of a data node).
	
	The procedure then continues until the last substripe $g'_{\alpha,1}$. In every step the amount of substripes read is the same as in the original HashTag code when recovering a data node.
\end{proof}

\section{Experiments in Hadoop}\label{Hadoop}
The repair duality discussed in Example \ref{Ex:Duality} of previous section was mainly influenced by one system characteristic: the access time for starting a read operation. In different environments of distributed storage systems there are several similar system characteristics that can affect the repair duality and its final optimal procedure. We next discuss this matter for Hadoop.

Hadoop is an open-source software framework used for distributed storage and processing of big data sets \cite{white2012hadoop}. From release 3.0.0-alpha2 Hadoop offers several erasure codes such as $(5,3)$, $(9, 6)$ and $(14, 10)$ Reed-Solomon (RS) codes. Hadoop Distributed File System (HDFS) has the concepts of \emph{Splits} and \emph{Blocks}. A Split is a logical representation of the data while a Block describes the physical alignment of data. Splits and Blocks in Hadoop are user defined: a logical split can be composed of multiple blocks and one block can have multiple splits. All these choices determine in a more complex way the access time for I/O operations.

To verify the performance of HashTag codes and their locally repairable and locally regenerating variants we implemented them in C/C++ and used them in HDFS. 

For the code $(9, 6)$ we used one NameNode, nine DataNodes, and one client node. All nodes had a size of 50 GB and were connected with a local network of 10 Gbps. The nodes were running on Linux machines equipped with Intel Xeon E5-2676 v3 running on 2.4 GHz. We have experimented with different block sizes (90 MB and 360 MB), different split sizes (512 KB, 1 MB and 4 MB) and different sub-packetization levels ($\alpha=1, 3, 6,$ and $9$) in order to check how they affect the repair time of one lost node. The measured times to recover one node are presented in Fig. \ref{9_6_Hadoop01}. Note that the sub-packetization level $\alpha=1$ represents the RS code that is available in HDFS, while for every other $\alpha = 3, 6, 9$ the codes are HashTag codes. In all measurements HashTag codes outperform RS. The cost of having significant number of I/O operations for the sub-packetization level $\alpha=9$ is the highest for the smallest block and split size (Block size of 90 MB and Split size of 512 KB). This is shown by yellow bars. As the split sizes increase, the disadvantage of bigger number of I/Os due to the increased sub-packetization diminishes, and the repair time decreases further (red and blue bars). 
\begin{figure}[h!]
	\includegraphics[width=3.6in]{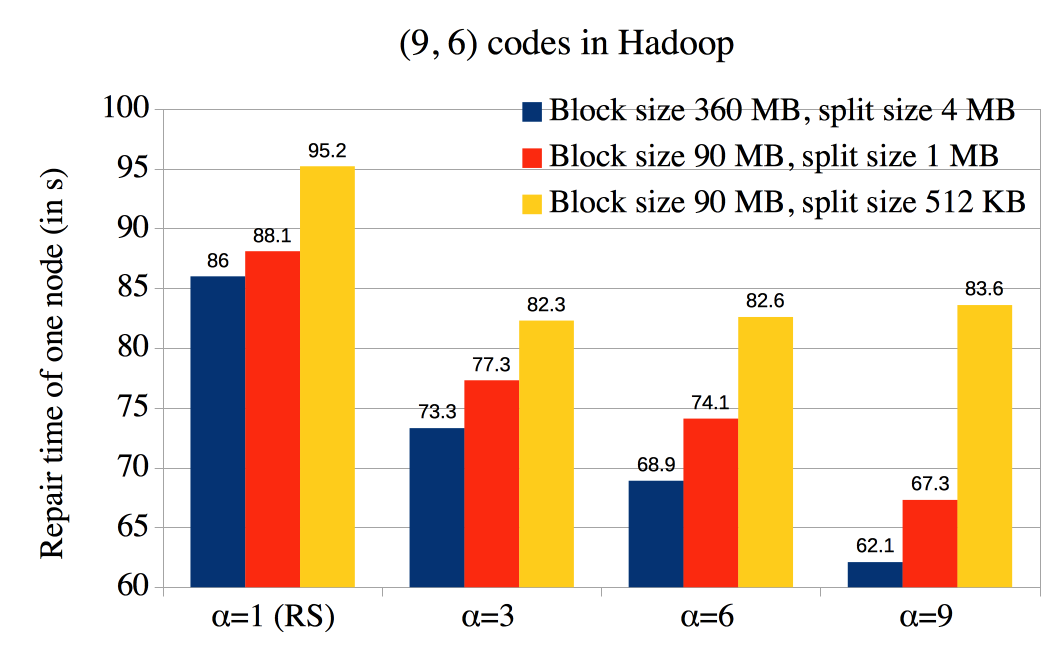}
	\caption{Experiments in HDFS. Time to repair one lost node of 50 GB with a $(9, 6)$ code for different sub-packetization levels $\alpha$. Note that the RS code for $\alpha=1$ is available in the latest release 3.0.0-alpha2 of Apache Hadoop.}
	\label{9_6_Hadoop01}
\end{figure}

In Fig. \ref{9_6_HadoopAlpha9}, we compare the repair times for one lost node of 50 GB with the codes from Examples \ref{Ex:HashTag0906alfa09-Split} and \ref{Ex:HashTag0906alfa09-Split3}. The cost of bigger redundancy with the locally repairable code $(10,6)$ (which is also a locally regenerative code since the sub-packetization level is $\alpha=9$) with $(l=2,\delta=2)$ (the red bar) is still not enough to outperform the ordinary $(9, 6)$ HashTag MSR code (the blue bar). However, paying even higher cost by increasing the redundancy for the other locally regenerative code $(11,6)$ with $(l=3,\delta=2)$ (the yellow bar) finally manages to outperform the repairing time for the ordinary $(9, 6)$ HashTag MSR code.
\begin{figure}[h!]
	\includegraphics[width=3.6in]{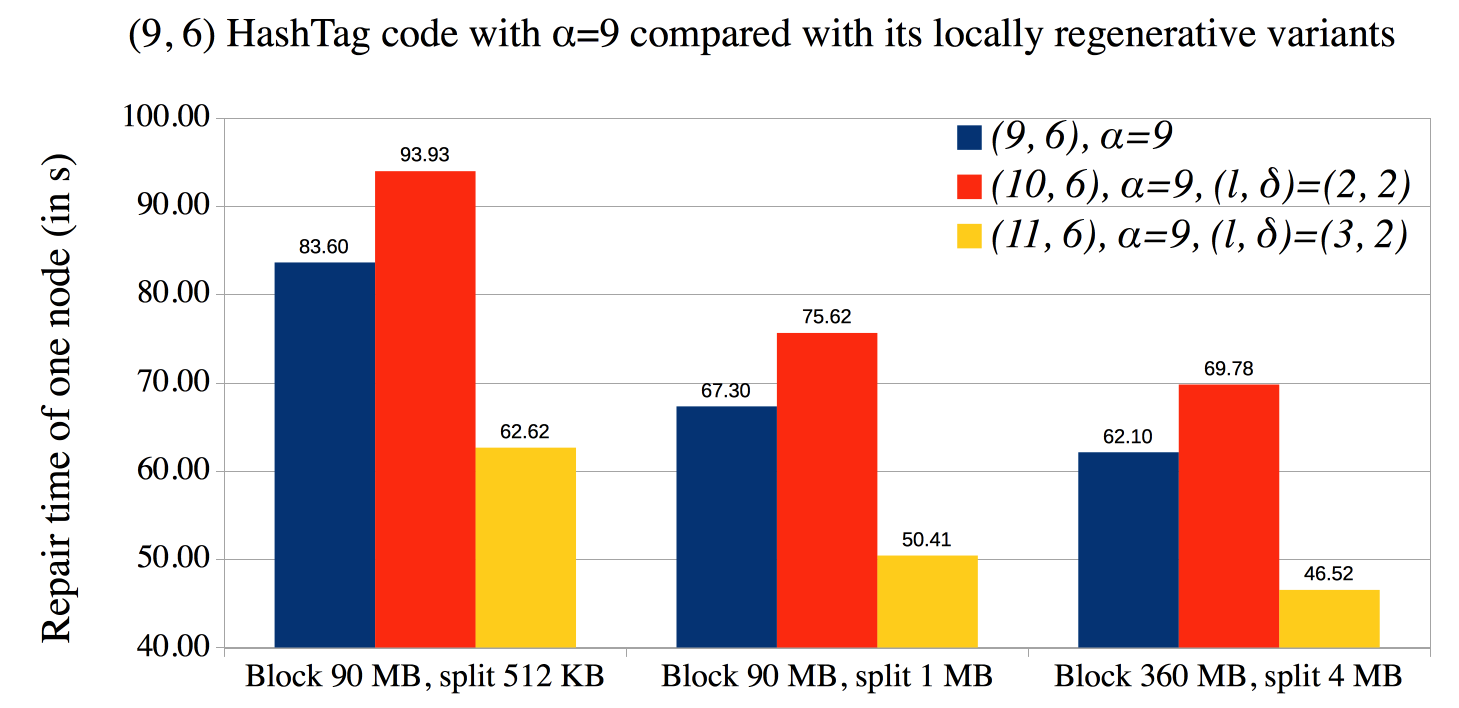}
	\caption{Comparison of repair times for one lost node of 50 GB for an ordinary $(9, 6)$ HashTag MSR code ($\alpha=9$), and its locally regenerative variants with $(l=2,\delta=2)$ and $(l=3,\delta=2)$.}
	\label{9_6_HadoopAlpha9}
\end{figure}

The situation of comparing the slightly less optimal $(9, 6)$ HashTag code with $\alpha=6$ with its locally repairable variants with $(l=2,\delta=2)$ and $(l=3,\delta=2)$ is different, and  this is presented in Fig. \ref{9_6_HadoopAlpha6}. In this case, all variants of locally repairable codes outperform the original HashTag code, i.e., they repair a failed node in shorter time than the original HashTag code from which they were constructed. 
\begin{figure}[h!]
	\includegraphics[width=3.6in]{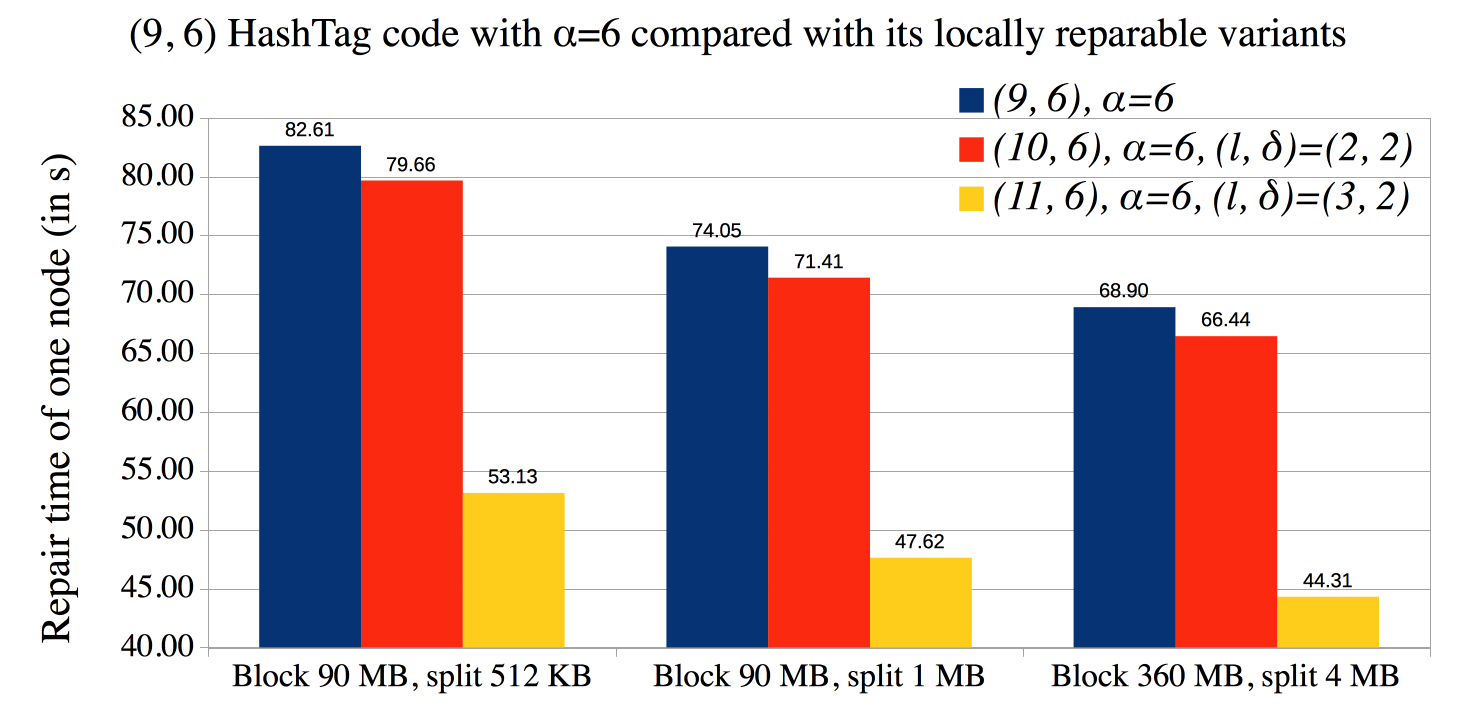}
	\caption{Comparison of repair times for one lost node of 50 GB for a $(9, 6)$ HashTag code with $\alpha=6$, and its locally repairable variants with $(l=2,\delta=2)$ and $(l=3,\delta=2)$.}
	\label{9_6_HadoopAlpha6}
\end{figure}

\section{Conclusions}\label{summary}
We constructed an explicit family of locally repairable and locally regenerating codes. We applied the technique of parity-splitting on HashTag codes and constructed codes with locality. For these codes we showed that there are two ways to repair a node (repair duality), and in practice which way is applied depends on optimization metrics such as the repair bandwidth, the number of I/O operations, the access time for the contacted parts and the size of the stored file. Additionally, we showed that the size of the finite field can be slightly lower than the theoretically obtained lower bound on the size in the literature.
We solved the problem of efficient repair of the global parities when the number of global parities is equal to the sub-packetization level.

\bibliography{refer}
\bibliographystyle{IEEEtran}

\end{document}